\documentclass[a4paper]{easychair}

\usepackage{natbib}

\usepackage{doc}

\title{Certified Branch-and-Bound MaxSAT Solving\\ (Extended Version)
}

\author{
	    Dieter Vandesande\inst{1,2} \and 
	Jordi Coll\inst{ 3} \and
	Bart Bogaerts\inst{ 2,1}
}

\institute{
Vrije Universiteit Brussel, Brussels, Belgium
\and
KU Leuven, Leuven, Belgium
\and
Universitat de Girona, Girona, Spain
 }

\authorrunning{Vandesande et al}

\titlerunning{Certified Branch-and-Bound MaxSAT Solving}

\newcommand\onlyAppendix[1]{#1}
\newcommand\notInAppendix[1]{}

\usepackage{refcount}
\newcommand{\setcounternextresult}[1]{\setcounter{theorem}{\getrefnumber{#1}-1}}
\usepackage[]{todonotes}
\usepackage{amsmath}
\usepackage{xspace}
\usepackage{soul}%highlighting
    \sethlcolor{yellow!50!white}%

\usepackage[capitalize]{cleveref}
\usepackage{amssymb}  % upharpoonright
\newcommand\m[1]{\ensuremath{#1}\xspace}
\newcommand\toolname[1]{\textsc{#1}\xspace}

\usepackage[linesnumbered,noend]{algorithm2e}
\usepackage[noend]{algorithmic}

\newcommand\limplies{\Rightarrow}
\newcommand\limpliedby{\Leftarrow}
\newcommand\lit{\m{\ell}}
\newcommand\var{\m{x}}
\newcommand\constraint{\m{C}}
\newcommand\degree{\m{A}}
\newcommand\assignment{\m{\alpha}}
\newcommand\assignmenttwo{\m{\beta}}
\newcommand\assignmentlocal{\m{\lambda}}
\newcommand{\Cdef}[5]{\m{\constraint_{\text{def}}^{#1}(#2, #3, #4, #5)}}
\newcommand{\Creif}[2]{\m{\constraint_{\text{reif}}^{#1}(#2)}} 

\newcommand\cost{\m{v}}
\newcommand\blocklit{\m{b}} 
\newcommand\objective{\m{\mathcal{O}}}

\newcommand\node\eta
\newcommand{\bddOf}[3]{\m{\mathrm{bdd}(#1,#2,#3)}} % variable - interval lower bound - interval upper bound
\newcommand{\mddOf}[3]{\m{\mathrm{mdd}(#1,#2,#3)}} % layer - interval lower bound - interval upper bound
\newcommand{\child}[2]{\m{\mathrm{child(#1, #2)}}} % node - literal true / else 
\newcommand\olnot[1]{\m{\overline{#1}}}

\newcommand\formula{\m{F}}

\newcommand\weight{\m{w}}
\newcommand\core{\m{K}}
\newcommand\weightedLocalCore{\m{q}} 
\newcommand\coretriple{\langle\weight,\reason,\core\rangle} 
\newcommand\coreset{\m{\mathcal{Q}}}
\newcommand\clause{C}
\newcommand\weightedLocalCoreClause{\m{\clause_{\weightedLocalCore}}}
\newcommand\coresetclause{\m{\clause_{\mathcal{Q}}}}

\newcommand\reason{\m{R}}
\newcommand\residualw{\mathit{res}}
\newcommand\comparison{\m{\bowtie}}
\newcommand{\litspartition}{\m{\mathcal{X}}}

\newcommand{\amo}{\m{\mathrm{AMO}}}
\newcommand{\amoUB}{\m{\mathrm{UB}}}

\newcommand\proofconfig{\m{\mathcal{F}}}

\newcommand\obj{\m{\mathcal{O}}}
\newcommand\weightInObj[1]{\m{w_{\obj}(#1)}}

\newcommand\bestval{\m{v^*}}
\newcommand\under[1]{\m{\!\!\upharpoonright\!\!_{#1}}}
\newcommand\subst[2]{\m{#1\under#2}}
\newcommand\witness\omega
\newcommand\sub\witness
\newcommand\lneg\lnot

\newcommand\maxcdcl{\toolname{MaxCDCL}}
\newcommand\veripb{\toolname{VeriPB}}

\usepackage{mathtools}

\newcommand\sumWithLongBottomLim[2]{{\mathop{\sum\mathrlap{#2}}_{#1}}}

\newcommand\eg{e.g.,\xspace}

\newcommand\false{\bot}
\newcommand\maxsat{MaxSAT\xspace}

\usepackage{amsthm}
\newtheorem{theorem}{Theorem}
\newtheorem{proposition}[theorem]{Proposition}
\newtheorem{definition}[theorem]{Definition}
\theoremstyle{note}
\newtheorem{example}[theorem]{Example}

\newcounter{appendixtheorem}

\newtheorem{aproposition}[appendixtheorem]{Proposition}
\newtheorem{adefinition}[appendixtheorem]{Definition}

\newcommand\codecall[1]{\texttt{#1}\xspace}

\newcommand\lfalse{\mathbf{f}}
\newcommand\ltrue{\mathbf{t}}
\newcommand\lequiv{\Leftrightarrow}

\newcommand\citeEG[1]{\cite[e.g.,][]{#1}}
\newcommand\citeWithNote[2]{\cite[#1][]{#2}}

\newcommand\ignore[1]{}

\usepackage{etoolbox}

\ExplSyntaxOn

\newcommand\setcitation[2]{%
	\csdef{mycommoncitation\text_uppercase:n{#1}}{#2}%
	\csappto{bbAllCommonCitations}{\cite{#2}\ }%%%for testing purpuses only	
}
\newcommand\getcitation[1]{%
	\csuse{mycommoncitation\text_uppercase:n{#1}}}

\newcommand\refto[1]{%
	\ifcsname  mycommoncitation\text_uppercase:n{#1}\endcsname%
	\getcitation{#1}%
	\else%
	#1%
	\fi%
}

\ExplSyntaxOff

\newcommand\mycite[1]{\cite{\refto{#1}}}

\setcitation{IDP}{DBBJD18PredicatelogicmodelinglanguageIDPsystem}
\setcitation{fodot}{DT08logicnonmonotoneinductivedefinitions}
\setcitation{CPsupport}{DBDD13ModelExpansionPresenceFunctionSymbolsUsing}
\setcitation{minisatid}{DBDD13ModelExpansionPresenceFunctionSymbolsUsing}
\setcitation{FunctionDetection}{DB13Detectionexploitationfunctionaldependenciesmodelgeneration}
\setcitation{lazyGrounding}{DDBS15LazyModelExpansionInterleavingGroundingSearch} 
\setcitation{Bootstrapping}{BJDJBD16BootstrappingInferenceIDPKnowledgeBaseSystem}
\setcitation{BootstrappingIDP}{BJDJBD16BootstrappingInferenceIDPKnowledgeBaseSystem}
\setcitation{FOSymmetry}{DBBD16localdomainsymmetrymodelexpansion}
\setcitation{FOLASP}{VDV21FOLASPFOInputLanguageAnswerSet}
\setcitation{inputster}{JJJ13CompilingInputFOinductivedefinitionsinto}
\setcitation{LearningPaper}{BBBDDJLR15Predicatelogicmodelinglanguagemodelingsolving}

\setcitation{foid}{DT08logicnonmonotoneinductivedefinitions}
\setcitation{cplogic}{VDB10FOIDextensionDLrules}
\setcitation{clog}{BVDV14FOCKnowledgeRepresentationLanguageCausality} %TODO replace by journal publication if one is published
\setcitation{foc}{BVDV14FOCKnowledgeRepresentationLanguageCausality}
\setcitation{inferenceClog}{BVDV14InferenceFOCModellingLanguage}
\setcitation{inferenceFOC}{BVDV14InferenceFOCModellingLanguage}
\setcitation{examplesClog}{BVDV14FOCRelatedModellingParadigms} %TODO replace by
\setcitation{satid}{MWDB08SATIDSatisfiabilityPropositionalLogicExtendedInductive}

\setcitation{fodot2asp}{DLTV19informalsemanticsAnswerSetProgrammingTarskian} %TODO replace by journal publication if one is publishedr
\setcitation{Tarskian}{DLTV19informalsemanticsAnswerSetProgrammingTarskian} %Same as the one above 
\setcitation{TarskianSemanticsASP}{DLTV19informalsemanticsAnswerSetProgrammingTarskian} %Same as the one above 
\setcitation{KBS}{DV08BuildingKnowledgeBaseSystemIntegrationLogic}
\setcitation{KBS-invitedtalk}{D16FOKnowledgeBaseSystemproject}
\setcitation{KBPE}{DWD11PrototypeKnowledge-BasedProgrammingEnvironment}

\setcitation{justifications}{DBS15FormalTheoryJustifications} 
\setcitation{justificationTheory}{DBS15FormalTheoryJustifications} 
\setcitation{AFT}{DMT00ApproximationsStableOperatorsWell-FoundedFixpointsApplications}

\setcitation{SAT}{BHvW21HandbookSatisfiability-SecondEdition}
\setcitation{HandbookOfSAT}{BHvW21HandbookSatisfiability-SecondEdition}
\setcitation{SMS}{KS24SATModuloSymmetriesGraphGenerationEnumeration}

\setcitation{totalizer}{BB03EfficientCNFEncodingBooleanCardinalityConstraints}

\setcitation{satsuma}{ABR24satsumaStructure-BasedSymmetryBreakingSAT}
\setcitation{BreakID}{DBBD16ImprovedStaticSymmetryBreakingSAT}
\setcitation{symchaff}{S09SymChaffexploitingsymmetrystructure-awaresatisfiabilitysolver}

\setcitation{saucy}{KSM10SymmetrySatisfiabilityUpdate}
\setcitation{dejavu}{AS21EngineeringFastProbabilisticIsomorphismTest}
\setcitation{nauty}{MP14PracticalgraphisomorphismII}
\setcitation{bliss}{JK07EngineeringEfficientCanonicalLabelingToolLarge}

\setcitation{DPLLT}{GHNOT04DPLLTFastDecisionProcedures}
\setcitation{SMT}{BSST21SatisfiabilityModuloTheories}

\setcitation{CP}{RvW06HandbookConstraintProgramming}
\setcitation{csp2asp}{DW11Translation-BasedConstraintAnswerSetSolving}
\setcitation{EZCSP}{B11Industrial-SizeSchedulingASPCP}
\setcitation{Inca}{DW12AnswerSetSolvingLazyNogoodGeneration}
\setcitation{Zinc}{MNRSGW08DesignZincModellingLanguage}
\setcitation{MiniZinc}{NSBBDT07MiniZincTowardsStandardCPModellingLanguage}

\setcitation{ASPComp2}{DVBGT09SecondAnswerSetProgrammingCompetition}
\setcitation{ASPComp3}{CIR14thirdopenanswersetprogrammingcompetition}
\setcitation{ASPComp4}{ACCDDIKK13FourthAnswerSetProgrammingCompetitionPreliminary}
\setcitation{ASPComp5}{CGMR16DesignresultsFifthAnswerSetProgramming}
\setcitation{ASPComp6}{GMR17SixthAnswerSetProgrammingCompetition}
\setcitation{ASPComp7}{GMR20SeventhAnswerSetProgrammingCompetitionDesign}
\setcitation{AspInPractice}{GKKS12AnswerSetSolvingPractice}
\setcitation{clasp}{GKS12Conflict-drivenanswersetsolvingFromtheory}
\setcitation{oclingo}{GGKOSS12StreamReasoningAnswerSetProgrammingPreliminary}
\setcitation{clingo}{GKKOSW16TheorySolvingMadeEasyClingo5}
\setcitation{gringo}{GST07GrinGoNewGrounderAnswerSetProgramming}
\setcitation{cmodels}{GLM04SAT-BasedAnswerSetProgramming}
\setcitation{DLV}{LPFEGPS06DLVsystemknowledgerepresentationreasoning}
\setcitation{GroundingBottleneck}{SPL14TwoApplicationsASP-PrologSystemDecomposablePrograms}
\setcitation{lazygroundingASP}{BW18ExploitingJustificationsLazyGroundingAnswerSet} 
\setcitation{justificationsAlpha}{BW18ExploitingJustificationsLazyGroundingAnswerSet}
\setcitation{ASP}{MT99StableModelsAlternativeLogicProgrammingParadigm}
\setcitation{OLL}{AKMS12Unsatisfiability-basedoptimizationclasp} %{MT99StableModelsAlternativeLogicProgrammingParadigm,N99LogicProgramsStableModelSemanticsConstraint,L99AnswerSetPlanning} 

\setcitation{Hilog}{CKW93HILOGFoundationHigher-OrderLogicProgramming}

\setcitation{MaxSAT}{LM21MaxSATHardSoftConstraints}

\setcitation{KR}{vLP08HandbookKnowledgeRepresentation}

\setcitation{lazyclausegeneration}{OSC09Propagationlazyclausegeneration}
\setcitation{FP}{H89ConceptionEvolutionApplicationFunctionalProgrammingLanguages}
\setcitation{GroundingWithBounds}{WMD10GroundingFOFOIDBounds}
\setcitation{GroundWithBounds}{WMD10GroundingFOFOIDBounds}
\setcitation{LTC}{BJBDVD14SimulatingDynamicSystemsUsingLinearTime}
\setcitation{SPSAT}{DBDDM12SymmetryPropagationImprovedDynamicSymmetryBreaking}
\setcitation{LCG}{S10LazyClauseGenerationCombiningPowerSAT}

\setcitation{GroundedFixpoints}{BVD15Groundedfixpointstheirapplicationsknowledgerepresentation}
\setcitation{PartialGroundedFixpoints}{BVD15PartialGroundedFixpoints}
\setcitation{LogicBlox}{GAK12LogicBloxPlatformLanguageTutorial}
\setcitation{proB}{LB08ProBautomatedanalysistoolsetBmethod}
\setcitation{NaturalInductions}{DV14Well-FoundedSemanticsPrincipleInductiveDefinitionRevisited} %TODO replace by journal publication if one is published
\setcitation{LP}{vK76SemanticsPredicateLogicProgrammingLanguage}

\setcitation{AF}{D95AcceptabilityArgumentsFundamentalRoleNonmonotonicReasoning}
\setcitation{ADF}{BW10AbstractDialecticalFrameworks}
\setcitation{ADFRevisited}{BSEWW13AbstractDialecticalFrameworksRevisited}
\setcitation{DefaultLogic}{R80LogicDefaultReasoning}
\setcitation{DL}{R80LogicDefaultReasoning}
\setcitation{AEL}{M85SemanticalConsiderationsNonmonotonicLogic}
\setcitation{minisat}{ES03ExtensibleSAT-solver}
\setcitation{completion}{C77NegationFailure}
\setcitation{ClarkCompletion}{C77NegationFailure}
\setcitation{wasp}{ADFLR13WASPNativeASPSolverBasedConstraint}
\setcitation{CEGAR}{CGJLV03Counterexample-guidedabstractionrefinementsymbolicmodelchecking}
\setcitation{CuttingPlane}{DFJ54SolutionLarge-ScaleTraveling-SalesmanProblem}
\setcitation{CuttingPlanes}{DFJ54SolutionLarge-ScaleTraveling-SalesmanProblem}
\setcitation{kodkod}{TJ07KodkodRelationalModelFinder}
\setcitation{CDCL}{SS99GRASPSearchAlgorithmPropositionalSatisfiability}
\setcitation{1UIP}{ZMMM01EfficientConflictDrivenLearningBooleanSatisfiability}
\setcitation{relevance}{JBDJD16RelevanceSATID}
\setcitation{relevance-implementation}{JDBJD16ImplementingRelevanceTrackerModule}
\setcitation{WFS}{VRS91Well-FoundedSemanticsGeneralLogicPrograms}
\setcitation{UnfoundedSet}{VRS91Well-FoundedSemanticsGeneralLogicPrograms}
\setcitation{UFS}{VRS91Well-FoundedSemanticsGeneralLogicPrograms}
\setcitation{stablesemantics}{GL88StableModelSemanticsLogicProgramming}
\setcitation{shatter}{ASM06EfficientSymmetryBreakingBooleanSatisfiability}
\setcitation{sbass}{DTW11Symmetry-breakinganswersetsolving}
\setcitation{lparsemanual}{url:lparsemanual}
\setcitation{AIC}{FGZ04Activeintegrityconstraints}
\setcitation{templates}{DvJD15Semanticstemplatescompositionalframeworkbuildinglogics}
\setcitation{templates2}{DvBJD16CompositionalTypedHigher-OrderLogicDefinitions}
\setcitation{sat-to-sat}{JTT16SAT-to-SATDeclarativeExtensionSATSolversNew}
\setcitation{sat-to-sat-qbf}{BJT16SolvingQBFInstancesNestedSATSolvers}
\setcitation{sat-to-sat-SO}{BJT16DeclarativeSolverDevelopmentCaseStudies}
\setcitation{XSB}{SW12XSBExtendingPrologTabledLogicProgramming}
\setcitation{KCmap}{DM02KnowledgeCompilationMap}
\setcitation{KnowledgeCompilationMap}{DM02KnowledgeCompilationMap}
\setcitation{TLA}{L02SpecifyingSystemsTLALanguageToolsHardware}
\setcitation{EventB}{A10ModelingEvent-B-SystemSoftwareEngineering}
\setcitation{MX}{MT05FrameworkRepresentingSolvingNPSearchProblems}
\setcitation{MIP}{S96Linearintegerprogramming-theorypractice}
\setcitation{PerfectModel}{P88DeclarativeSemanticsDeductiveDatabasesLogicPrograms}
\setcitation{SafeInductions}{BVD17SafeInductionsAlgebraicStudy}
\setcitation{alpha}{W17BlendingLazy-GroundingCDNLSearchAnswer-SetSolving}
\setcitation{omiga}{DEFWW12OMiGAOpenMindedGroundingOn-The-FlyAnswer}
\setcitation{gasp}{PDPR09GASPAnswerSetProgrammingLazyGrounding}
\setcitation{asperix}{LN09FirstVersionNewASPSolverASPeRiX}
\setcitation{CTL}{CE81DesignSynthesisSynchronizationSkeletonsUsingBranching-Time}
\setcitation{AFT-AIC}{BC18Fixpointsemanticsactiveintegrityconstraints}
\setcitation{UltimateApproximator}{DMT04Ultimateapproximationapplicationnonmonotonicknowledgerepresentation}
\setcitation{KripkeKleene}{F85Kripke-KleeneSemanticsLogicPrograms}
\setcitation{AFT-HO}{CRS18ApproximationFixpointTheoryWell-FoundedSemanticsHigher-Order} 
\setcitation{HereThere}{H30formalenRegelnintuitionistischenLogik}
\setcitation{dAEL}{VCBD16DistributedAutoepistemicLogicApplicationAccessControl}
\setcitation{SDD}{D11SDDNewCanonicalRepresentationPropositionalKnowledge}
\setcitation{HEX}{EIST06dlvhexProverSemantic-WebReasoningunderAnswer-Set}
\setcitation{wADF}{BSWW18WeightedAbstractDialecticalFrameworks}
\setcitation{wADFfix}{BPSWW18WeightedAbstractDialecticalFrameworksExtendedRevised}
\setcitation{TransitionSystems}{NOT06SolvingSATSATModuloTheoriesFrom}
\setcitation{galliwasp}{MG12GalliwaspGoal-DirectedAnswerSetSolver}
\setcitation{clingcon}{BKOS17Clingconnextgeneration}
\setcitation{lp2sat}{JN11CompactTranslationsNon-disjunctiveAnswerSetPrograms}
\setcitation{lp2mip}{LJN12AnswerSetProgrammingMixedIntegerProgramming}
\setcitation{lp2diff}{JNS09ComputingStableModelsReductionsDifferenceLogic}
\setcitation{lp2acyc}{GJR14AnswerSetProgrammingSATmoduloAcyclicity}
\setcitation{PB}{RM09Pseudo-BooleanCardinalityConstraints}
\setcitation{CuttingPlanes}{CCT87complexitycutting-planeproofs}
\setcitation{RoundingSAT}{EN18DivideConquerTowardsFasterPseudo-BooleanSolving}
\setcitation{PRS}{DG02InferenceMethodsPseudo-BooleanSatisfiabilitySolver}
\setcitation{sat4j}{LP10Sat4jlibraryrelease22}
\setcitation{mingo}{LJN12AnswerSetProgrammingMixedIntegerProgramming}
\setcitation{pbmodels}{LT05Pbmodels-SoftwareComputeStableModels}
\setcitation{HEF-LP}{GLL07Head-Elementary-Set-FreeLogicPrograms}
\setcitation{HCF-LP}{BD94PropositionalSemanticsDisjunctiveLogicPrograms}
\setcitation{aspcore2}{CFGIKKLM20ASP-Core-2InputLanguageFormat}
\setcitation{SemanticWeb}{AGvH12SemanticWebPrimer3rdEdition}
\setcitation{QBF}{KB21TheoryQuantifiedBooleanFormulas}
\setcitation{SMUS}{IPLM15SmallestMUSExtractionMinimalHittingSet}
\setcitation{CDCLsym}{MBCK18CDCLSymIntroducingEffectiveSymmetryBreakingSAT}
\setcitation{SEL}{DBB17SymmetricExplanationLearningEffectiveDynamicSymmetry}
\setcitation{Coq}{BC04InteractiveTheoremProvingProgramDevelopment-}
\setcitation{Isabelle}{NPW02IsabelleHOL-ProofAssistantHigher-OrderLogic}
\setcitation{HOL}{SN08BriefOverviewHOL4}
\setcitation{lean}{dU21Lean4TheoremProverProgrammingLanguage}
\setcitation{CakeML}{TMKFON19verifiedCakeMLcompilerbackend}
\setcitation{FRAT}{BCH22FlexibleProofFormatSATSolver-ElaboratorCommunication}
\setcitation{DRUP}{GN03VerificationProofsUnsatisfiabilityCNFFormulas}
\setcitation{RUP}{GN03VerificationProofsUnsatisfiabilityCNFFormulas}
\setcitation{GRIT}{CMS17EfficientCertifiedResolutionProofChecking}
\setcitation{TraceCheck}{url:TraceCheck}
\setcitation{LRAT}{CHHKS17EfficientCertifiedRATVerification}

\setcitation{PR}{HKB17ShortProofsWithoutNewVariables}
\setcitation{DRAT}{WHH14DRAT-trimEfficientCheckingTrimmingUsingExpressive}
\setcitation{satcomp2016}{BHJ17SATCompetition2016RecentDevelopments}
\setcitation{satcomp2013}{url:SATcomp2013}

\setcitation{Essence}{FHJHM08Essenceconstraintlanguagespecifyingcombinatorialproblems}
\setcitation{smodels}{SN01SmodelsSystem}
\setcitation{lparse}{S98ImplementationLocalGroundingLogicProgramsStable}
\setcitation{IDPZ3}{CVVD22IDP-Z3reasoningengineFO} %TODO replace by real publication
\setcitation{IDP-Z3}{CVVD22IDP-Z3reasoningengineFO} %} %TODO replace by real publication
\setcitation{Z3}{dB08Z3EfficientSMTSolver}

\setcitation{DMN}{DMN} %TODO is there something better than just this on-line thing?
\setcitation{Planning}{GNT04Automatedplanning-theorypractice}
\setcitation{AutomatedPlanning}{GNT04Automatedplanning-theorypractice}
\setcitation{RC2}{IMM19RC2EfficientMaxSATSolver}
\setcitation{enfragmo}{phd/Aavani14}
\setcitation{TPTP}{url:tptp}

\renewcommand{\cite}[1]{\citep{#1}\xspace}
\providecommand\citeEG[1]{\cite{#1}\xspace}
\renewcommand\citeEG[1]{\cite{#1}\xspace}
\renewcommand\citeWithNote[2]{(#1 \citet{#2})}
\begin{document}
	
	\maketitle

	\begin{abstract}
		Over the past few decades, combinatorial solvers have seen remarkable performance improvements, enabling their practical use in real-world applications.
In some of these applications, ensuring the correctness of the solver's output is critical.
However, the complexity of modern solvers makes them susceptible to bugs in their source code.
In the domain of satisfiability checking (SAT), this issue has been addressed through \emph{proof logging}, where the solver generates a formal proof of the correctness of its answer.
For more expressive problems like MaxSAT, the optimization variant of SAT, proof logging had not seen a comparable breakthrough until recently.

In this paper, we show how to achieve proof logging for state-of-the-art techniques in Branch-and-Bound MaxSAT solving. 
This includes certifying look-ahead methods used in such algorithms as well as advanced clausal encodings of pseudo-Boolean constraints based on so-called Multi-Valued Decision Diagrams (MDDs).
We implement these ideas in MaxCDCL, the dominant branch-and-bound solver, and experimentally demonstrate that proof logging is feasible with limited overhead, while proof checking remains a challenge. 

\onlyAppendix{ 
	This is an extended version of a paper that will be published in the proceedings of AAAI 2026 \cite{VCB26CertifiedBranch-and-BoundMaxSATSolving}. 
	It extends the published version with a technical appendix including proofs and more details. }

	\end{abstract}
	
	
	%\bart{SOMEWHERE add ref to MDD explanations paper, see first reviewer AAAI} 
	%\bart{RELATED WORK add refs to maxsat resolution, proof builder for maxsat see AAAI review of Py} 
	
	\section{Introduction}\label{sec:introduction}\label{sec:intro}
	% !TeX root = ../CertifiedBranch-And-Bound.tex
With increasingly efficient solvers being developed across various domains of combinatorial search and optimization, we have effectively reached a point where NP-hard problems are routinely addressed in practice.
This maturation has led to the deployment of solvers in a wide range of real-world applications, including safety-critical systems and life-impacting decision-making processes—such as 
verifying software for transportation infrastructure \cite{FS04AutomaticVerificationSafetyRulesSubwayControl}, 
checking the correctness of plans for the space shuttle's reaction control system \cite{NBGWB01A-PrologDecisionSupportSystemSpaceShuttle}, 
and matching donors with recipients in kidney transplants \cite{MO14PairedAltruisticKidneyDonationUKAlgorithms}.
Given these high-stakes applications, it is crucial that the results produced by such solvers are guaranteed to be correct.
Unfortunately, this is not always the case: numerous reports have documented solvers producing infeasible solutions or incorrectly claiming optimality or unsatisfiability \cite{BBNOV23CertifiedCore-GuidedMaxSATSolving,BLB10AutomatedTestingDebuggingSATQBFSolvers,CKSW13hybridbranch-and-boundapproachexactrationalmixed-integer,GSD19SolverCheckDeclarativeTestingConstraints,JHB12InprocessingRules}.

This problem calls for a systematic solution.
The most straightforward approach would be to use \emph{formal verification}; that is, employing a proof assistant \cite{BC04InteractiveTheoremProvingProgramDevelopment-,NPW02IsabelleHOL-ProofAssistantHigher-OrderLogic,SN08BriefOverviewHOL4,dU21Lean4TheoremProverProgrammingLanguage} to formally prove the correctness of a solver.
However, in practice, formal verification often comes at the cost of performance \cite{phd/Fleury20,url:satcomp2022}, which is precisely what has driven the success of combinatorial optimization.

Instead, we advocate for the use of \emph{certifying algorithms} \cite{MMNS11Certifyingalgorithms}, an approach also known as \emph{proof logging} in the context of combinatorial optimization.
With proof logging, a solver not only produces an answer (\eg an optimal solution to an optimization problem) but also a \emph{proof of correctness} for that answer. 
This proof can be checked efficiently (in terms of the proof size)  by an independent tool known as a \emph{proof checker}.
Beyond ensuring correctness, proof logging also serves as a powerful software development methodology: it supports advanced \emph{testing} and \emph{debugging} of solver implementations \cite{BB09FuzzingDelta-DebuggingSMTSolvers}.
Moreover, the generated proofs constitute an auditable trail and can be used to speed-up the generation of explanations of why a particular conclusion was reached \cite{BFDBG26UsingCertifyingConstraintSolversGeneratingStep-wise}.

Proof logging was pioneered in the field of Boolean satisfiability (SAT), where numerous proof formats and proof checkers, including formally verified ones, have been developed over the years \cite{url:tracecheck,GN03VerificationProofsUnsatisfiabilityCNFFormulas,HHW13Trimmingwhilecheckingclausalproofs,WHH14DRAT-trimEfficientCheckingTrimmingUsingExpressive,CHHKS17EfficientCertifiedRATVerification}.
%A major breakthrough occurred when proof logging was used to resolve the Pythagorean triple problem, resulting in what became known as the ``largest math proof ever'' \cite{HKM16SolvingVerifyingBooleanPythagoreanTriplesProblem}.
%Furthermore, f
For several years, proof logging has been mandatory in the main tracks of the annual SAT solving competition, reflecting the community's strong commitment to ensuring that all SAT solvers are certifying.

In this paper, we are concerned with \emph{maximum satisfiability} (\maxsat), the optimization variant of SAT. 
In contrast to SAT, proof logging in \maxsat is still relatively uncommon.
While several proof systems for \maxsat have been developed 
\cite{%
	BLM07ResolutionMax-SAT,LNOR11FrameworkCertifiedBooleanBranch-and-BoundOptimization,MM11ValidatingBooleanOptimizers,PCH20TowardsBridgingGapBetweenSATMax-SAT}
and a solver has been designed specifically to generate such proofs \cite{PCH22ProofsCertificatesMax-SAT}, only with the recent advent of the \veripb proof system~\cite{BGMN23CertifiedDominanceSymmetryBreakingCombinatorialOptimisation,GN21CertifyingParityReasoningEfficientlyUsingPseudo-Boolean} did a general-purpose proof logging methodology for \maxsat see the light \cite{VDB22QMaxSATpbCertifiedMaxSATSolver}.
\veripb, a proof system for linear inequalities over 0--1 integer variables, has since been successfully applied to \maxsat solvers based on \emph{solution-improving search} \cite{VDB22QMaxSATpbCertifiedMaxSATSolver,msc/Vandesande23,BBNOPV24CertifyingWithoutLossGeneralityReasoningSolution-Improving} (where a SAT oracle is repeatedly queried to find solutions that improve upon the current best), as well as to solvers using \emph{core-guided search} \cite{BBNOV23CertifiedCore-GuidedMaxSATSolving} (where a SAT oracle is queried iteratively under increasingly relaxed optimistic assumptions).
The techniques developed there have also led to certification of generalizations of MaxSAT for multi-objective problems \cite{JBBJ25CertifyingParetoOptimalityMulti-ObjectiveMaximumSatisfiability}.

In addition to solution-improving and ore-guided search, state-of-the-art \maxsat solvers also employ other strategies, such as \emph{implicit hitting set} and \emph{branch-and-bound} search.
In a companion paper \cite{IVSBBJ26EfficientReliableHitting-SetComputationsImplicitHitting}, we show how to certify IHS search, and in the current paper, 
we develop \veripb-based certification for branch-and-bound search \cite{LM21MaxSATHardSoftConstraints, LXCMHH22Boostingbranch-and-boundMaxSATsolversclauselearning, jsat/AbrameH14, sat/Kugel10, LMP07NewInferenceRulesMax-SAT, jair/HerasLO08}, thereby finally covering all key search paradigm in \maxsat solving.
Modern branch-and-bound solvers combine conflict-driven clause learning (CDCL) \mycite{CDCL} with a sophisticated bounding function that determines whether
% the current search node  is ``hopeless'', in the sense that 
it is the case that
no assignment refining the current node can improve upon the best solution found so far \cite{LXCMHH22Boostingbranch-and-boundMaxSATsolversclauselearning, LXCMHH21CombiningClauseLearningBranchBoundMaxSAT, COLL2025107195}. 
While it is well understood how to certify CDCL search using \veripb, the main challenge lies in certifying the conclusions drawn by the bounding function.
This bounding function employs look-ahead methods to generate (conditional) unsatisfiable cores: sets of literals that cannot simultaneously be true.
These cores are then combined to estimate the best possible objective value that remains achievable.
%In the unweighted case in particular, this combination process relies on subtle reasoning, which we explicitly formalize and encode within the \veripb proof system.

In order to bring this certification to practice, there are more hurdles to overcome: state-of-the-art branch-and-bound solvers employ several other clever tricks to speed up the solving process, such as pre-processing methods as well as integrating ideas from other search paradigms. 
One particular technique that proved to be challenging is the use of \emph{multi-valued decision diagrams} (MDDs) in order to create a CNF encoding of a solution-improving constraint (which is enabled or disabled heuristically depending  on the size of the instance at hand) \cite{BCSV20MDD-basedSATencodingpseudo-Booleanconstraintsat-most-one}. 
This encoding generalizes the encoding of \citet{ANORM12NewLookBDDsPseudo-BooleanConstraints} based on Binary Decision Diagrams (BDDs) by allowing splits on sets of variables instead of a single variable.
While MDDs have been used in a lazy clause generation setting \cite{GSS11MDDpropagatorsexplanation}, their clausal encoding has never been certified.  
From the perspective of proof logging, one main challenge with BDD- or MDD-based encodings is that multiple (equivalent) constraints are represented by the same variable in this encoding. 
In the proof, we then have to show that they are indeed  equivalent. 
As an example, consider the constraints 
\begin{align*}
	12x_1 + 5 x_2 + 4x_3 & \geq 6\text{\qquad and} \\
	12x_1 + 5 x_2 + 4x_3 &\geq 9. 
\end{align*}
Taking into account that the variables $x_i$ take values in $\{0,1\}$, it is not hard to see that these constraints are equivalent: 
all combinations of truth assignments that lead to the left-hand side taking a value at least six, must assign it a value of at least nine. 
In other words: the left-hand side cannot take values $6$, $7$, or $8$. 
In general, checking whether such a linear expression can take a specific value (e.g., $7$) is well-known to be NP hard, but BDD (and MDD) construction algorithms will in several cases detect this efficiently. % (in the worst-case this can take exponential time, which is why such encodings will only be enabled for constraints with few variables and small coefficients).
The main question of interest for us is how to convince a proof checker of the fact that these two constraints are equivalent, without doing a substantial amount of additional work. 
We achieve this using an algorithm that proves this property for all nodes in an MDD in a linear pass over its representation.

To evaluate our approach, we add proof logging to the branch-and-bound solver WMaxCDCL~\cite{COLL2025107195}. WMaxCDCL is the current state-of-the-art in branch-and-bound MaxSAT solving, which is showcased by the fact that it won (as part of a portfolio-solver) the unweighted track of the 2024 edition and the weighted track of the 2023 edition \cite{url:MSE23, url:MSE24}.
In the experiments, we focus on evaluating the overhead in the solving time, while also measuring the time necessary to check the produced proofs; the results show that proof logging is indeed possible without significant overhead for most cases, while proof checking overhead can be improved.

%% Maybe need another paragraph summarizing our contributions? 
%Summarized our main contributions can be summarized as follows: 
%\begin{enumerate}
%	\item We show how look-ahead--based bounding can be made certifying. In combination with well-known proof logging methods, this then yields proof logging for branch-and-bound \maxsat solving. 
%	\item We show how other methods, in particular MDD-based encoding of pseudo-Boolean constraints 
%	\item We experimentally validate thi
%\end{enumerate}

The rest of this paper is structured as follows. 
In \cref{sec:prelims}, we recall some preliminaries about \maxsat solving and \veripb-based proof logging. 
\cref{sec:maxcdcl} is devoted to presenting the core of the \maxcdcl algorithm, focusing on look-ahead--based bounding, as well as explaining how to get a certifying version of this. 
In \cref{sec:MDD}, we explain how BDDs and MDDs are used to encode PB constraints and how this can be made certifying. 
\cref{sec:experiments} contains our experimental analysis and \cref{sec:concl} concludes the paper. 

%\nonExtended{
%	Due to page restrictions, 
Formal details and proofs are 
\onlyAppendix{can be found in the appendices of this document.
 First, the appendix it contains a short discussion on how to use the redundance-based strengthening rule for proving reification and proofs by contradiction. Next, for Section~\ref{sec:maxcdcl} on Certifying Branch-and-Bound with Look-Ahead, it includes the proof for Theorem~\ref{thm:proofmain}, together with a more detailed example that shows the look-ahead call that leads to the discovery of the local cores for which clause $\coresetclause$ are derived. For Section~\ref{sec:mdd} on Certifying CNF Encodings based on Binary (and Multi-Valued) Decision Diagrams, it contains more information on the MDD encoding, together with all the formal details on how to certify MDD-based encodings for the solution-improving constraint. We also include the proofs that were omitted in the paper due to space restrictions. Importantly, the proofs presented in this technical appendix are of constructive nature on purpose: it shows exactly how to perform the derivation in the VeriPB proof system, showing how it is implemented in the solver.}
\notInAppendix{often omitted, but can be found in the extended version of this paper \cite{vandesande2025certifiedbranchandboundmaxsatsolving}.}
%	\bart{Can we cite ``full version on ArXiV'' here? Dieter can you prep one?}

%}

	\section{Preliminaries}\label{sec:prelims}
	We first recall some concepts from pseudo-Boolean optimization and \maxsat solving. 
Afterwards, we introduce the \veripb proof system. 
A more complete exposition can be found in previous work  \cite{\refto{Maxsat},BGMN23CertifiedDominanceSymmetryBreakingCombinatorialOptimisation}.

\subsection{Pseudo-Boolean Optimization and MaxSAT}

In this paper, all variables are assumed to be \emph{Boolean}; meaning they take a value in  $\{0,1\}$. 
A \emph{literal} $\lit$ is a Boolean variable \var or its negation $\olnot{\var}$. 
A \emph{pseudo-Boolean (PB) constraint} \constraint is a $0$--$1$ integer linear inequality
$ \sum_i w_i \lit_i \comparison \degree$, with $w_i, \degree$ integers, $\lit_i$ literals and $\comparison \in \{\geq, \leq\}$.  
Without loss of generality, we will often assume our constraints to be \emph{normalized}, meaning that the $\lit_i$ are different literals, the comparison symbol $\comparison$ is equal to $\geq$ and all coefficients $w_i$ and the degree \degree are non-negative. 
%%% NOTE! One reviewer claims that $lit_i$ should be different VARIABLES. But this is wrong. Do not fix. 
A \emph{formula} is a conjunction of PB constraints. 
A \emph{clause} is a special case of a PB constraint having all $w_i$ and \degree equal to one. 
%A \emph{cardinality constraint} is a PB constraint where all $w_i$ are one. 
%If $L$ is a set of literals, we will sometimes simply write $L$ as a shorthand for $\sum_{\lit \in L}\lit$ and thus write constraints such as $L + 3 K \geq 42$, meaning  $\sum_{\lit \in L}\lit + \sum_{\lit \in K}3\cdot \lit \geq 42$. 

A \emph{(partial) assignment} $\assignment$ is a (partial) function from the set of variables to $\{0, 1\}$; it is extended to literals in the obvious way, and we sometimes identify an assignment with the set of literals it maps to $1$.
We write \mbox{$\constraint\under\assignment$} for the constraint obtained from $\constraint$  by substituting all assigned variables $\var$ by their assigned value $\assignment(\var)$. 
A constraint $\constraint$ is \emph{satisfied} under $\assignment$ if 
$\sum_{\assignment(\lit_i) =  1} w_i \geq \degree$. 
A formula $\formula$ is satisfied under $\assignment$ if all of its constraints are. 
We say that $\formula$ \emph{implies} $\constraint$ (and write $\formula\models\constraint$) if all assignments that satisfy $\formula$ also satisfy $\constraint$.  %\bart{last sentence needed?}
A constraint $\constraint$ \emph{(unit) propagates} a literal $x$ to a value in $\{0,1\}$ under a (partial) assignment $\assignment$ if assigning $x$ the opposite value makes the constraint $\constraint\under\assignment$ unsatisfiable. 

A \emph{pseudo-Boolean optimization} instance consists of a formula $\formula$ and a linear term $\objective = \sum_i \cost_i \blocklit_i$ (called the \emph{objective}) to be minimized, where $\cost_i$ are integers (without loss of generality assumed to be \emph{positive} here) and $\blocklit_i$ are different literals, which we will refer to as \emph{objective literals}. 
If $\lit$ is a literal, we write $\weightInObj{\lit}$ for the weight of $\lit$ in $\objective$, i.e., $\weightInObj{\blocklit_i}=\cost_i$ for all $i$ and $\weightInObj{\lit}=0$ for all other literals.
In line with previous work on certifying MaxSAT solvers \citeEG{BBNOV23CertifiedCore-GuidedMaxSATSolving,msc/Vandesande23,BBNOPV24CertifyingWithoutLossGeneralityReasoningSolution-Improving}, in this paper, we view \maxsat as the special case of pseudo-Boolean optimization where $\formula$ is a conjunction of clauses. A discussion on the equivalence of this view and the more clause-centric view using hard and soft clauses is given by for example \citet{ecai/LeivoBJ20}.

\subsection{The \veripb Proof System}

For a pseudo-Boolean optimization instance $(\formula,\objective)$, the \veripb proof system \cite{BGMN23CertifiedDominanceSymmetryBreakingCombinatorialOptimisation,GN21CertifyingParityReasoningEfficientlyUsingPseudo-Boolean} maintains a \emph{proof configuration} $\proofconfig$, which is a set of constraints derived so far, and is initialized with $\formula$.
%$(\coresetVeriPB,\derivedsetVeriPB)$ where 
%$\coresetVeriPB$ (standing for \emph{core})  and $\derivedsetVeriPB$ (standing for \emph{derived}) are sets of constraints (initialized as $\formula$ and $\emptyset$, respectively). Constraints can be moved from $\derivedsetVeriPB$ to $\coresetVeriPB$ but not vice versa.
We are allowed to update the configuration using the cutting planes proof system \mycite{CuttingPlanes}:
\begin{description}
	\item [Literal Axioms:] For any literal, we can add  $\lit_{i} \geq 0$ to $\proofconfig$.
	\item [Linear Combination:] Given two PB constraints $C_1$ and $C_2$ in $\proofconfig$, we can add a positive linear combination of $C_1$ and $C_2$ to $\proofconfig$.
	\item [Division:] Given the normalized PB constraint $ \sum_i w_i \lit_i \geq \degree$ in $\proofconfig$ and a positive integer $c$, we can add the constraint $ \sum_i \lceil w_i/c \rceil \lit_i \geq \lceil
	\degree/c \rceil$ to $\proofconfig$.
	\item [Saturation:]  Given the normalized PB constraint $ \sum_i a_i l_i \geq A$ in $\proofconfig$, we can add $\sum_i b_i l_i \geq A$ with $b_i=\min(a_i,A)$ for each $i$, to $\proofconfig$.
\end{description}
 Importantly, in many cases, it is not required to show precisely why a constraint can be derived, but the verifier can figure it out itself by means of so-called \textbf{reverse unit propagation (RUP)}
\cite{GN03VerificationProofsUnsatisfiabilityCNFFormulas}:  if applying unit propagation until fixpoint on $\proofconfig \land \lnot \constraint$ propagates to a contradiction, then we can add $\constraint$ to $\proofconfig$. 

% Importantly, in many cases, it is not required to show precisely why a constraint can be derived, but the verifier can figure it out itself by means of so-called \emph{reverse unit propagation}
% \cite{GN03VerificationProofsUnsatisfiabilityCNFFormulas}. 
% \dieter{Check if you can write RUP inline, and maybe also put proof by contradiction here. Next talk about RBS and reification. }
%\begin{description}
%	\item[RUP] If applying unit propagation until fixpoint on $\proofconfig \land \lnot \constraint$ propagates to a contradiction, then we can add $\constraint$ to $\proofconfig$.
%\end{description}

\veripb has some additional rules, which are briefly discussed below. We refer to earlier work~\cite{BGMN23CertifiedDominanceSymmetryBreakingCombinatorialOptimisation} for more details.
There is a rule for dealing with optimization statements:
\begin{description}
	\item [Objective Improvement:] Given an assignment $\assignment$ that satisfies $\proofconfig$, we can add the constraint $\objective < \objective \under \assignment $ to $\proofconfig$.
\end{description}
The constraint added to $\proofconfig$ is also called the solution-improving constraint; if $0\geq 1$ can be derived from $\proofconfig$ after applying the objective improvement rule, then it is concluded that $\alpha$ is an optimal solution.
%It is also possible to rewrite the objective:
%\begin{description}
%	\item[Objective Reformulation] Given a linear term $\objective'$ in $\objective$,  $\objective$ can be rewritten by replacing $\objective'$ with $\objective'_\text{new}$ if we have shown that $\objective' \geq \objective'_\text{new}$ and $\objective' \leq \objective'_\text{new}$.
%\end{description}

\veripb allows deriving non-implied constraints with the \emph{redundance-based strengthening rule}:
a generalization of the RAT rule~\cite{JHB12InprocessingRules} and commonly used in proof systems for SAT. 
In this paper, the redundance rule is mainly used for \textbf{reification}: for any PB constraint $C$ and any fresh variable $v$, not used before, two applications of the redundance rule allow us to derive PB constraints that express $v\limplies C$ and $v\limpliedby C$. 
Moreover, the redundance rule allows us to derive implied constraints by a \textbf{proof by contradiction}: if we have a cutting planes derivation that shows that $\proofconfig \land \lneg \constraint \models 0 \geq 1$, then we can add $\constraint$ to $\proofconfig$.\footnote{In practice, this is mimicked by applying the \emph{redundance-based strengthening rule} with an empty witness; details can be found in the supplementary material.} %no space before footnote please. 

	\section{Certifying Branch-And-Bound with Look-Ahead} \label{sec:maxcdcl}
	In this section, we present the working of modern branch-and-bound search for MaxSAT, specifically for the \maxcdcl algorithm. 
Our presentation is given with proof logging in mind (we immediately explain what is essential for each step from the proof logging perspective) and hence we often take a different perspective from the original papers introducing this algorithm \cite{LXCMHH21CombiningClauseLearningBranchBoundMaxSAT, COLL2025107195}.

\maxcdcl combines branch-and-bound search and conflict-driven clause learning (CDCL). 
An overview of the algorithm is given in Algorithm~\ref{alg:maxcdcl}. 
As can be seen, this essentially performs CDCL search (this is the \emph{branching} aspect), using standard techniques such as assigning variables, unit propagation, conflict analysis, and non-chronological backtracking \mycite{CDCL}. 

\maxcdcl maintains the objective value $\bestval$ of the best found solution so far, which is initialized as $+\infty$. 
However, after unit propagation the \texttt{lookahead} procedure is called. 
This procedure is where the \emph{bounding} aspect happens: here sophisticated techniques are used to determine whether the current assignment can still be extended to a satisfying assignment that improves upon the current best objective value. If not, the search is interrupted and a new clause forcing the solver to backtrack is learned.

It is well known how to get 
\veripb-based certification for CDCL \citeWithNote{this is done for instance in a MaxSAT setting by}{VDB22QMaxSATpbCertifiedMaxSATSolver},
% Not sure if this warrants this many references; we could also search for the first reference where this is explicitly done. 
so we focus here only on how to get proof logging for the \codecall{lookahead} method.\footnote{Obviously, our implementation with proof logging also handles the other aspects of the algorithm.} %no space before footnote please. 

\begin{algorithm}
	\caption{Overview of the  \maxcdcl algorithm for branch-and-bound \maxsat solving. }\label{alg:maxcdcl}
	\KwIn{CNF formula $\formula$, objective $\objective$}
	\KwOut{optimal assigment or \textbf{UNSAT}}
$\formula \gets$ \codecall{Preprocess}($\formula$)\;
$\assignment \gets \emptyset$;\quad $\assignment^*\gets \textbf{UNSAT}$; \quad $\bestval\gets +\infty$\;
\While{true}{
	\codecall{unit-propagate}($\formula$, $\assignment$)\;
	  \hl{\mbox{\codecall{lookahead}($\formula$, $\objective$, $\bestval$, $\assignment$)}}\;
	\If{ conflict detected }{
		\If{ at root level}{
			\Return $\assignment^*$ \label{line:returnmaxcdcl}
		}
		Analyze conflict and backtrack\; 
	}
	\ElseIf{all variables assigned}{
		\lIf{$\objective\under\assignment < \bestval$ }{$\assignment^*\gets \assignment$;\quad $\bestval \gets \objective\under\assignment$\label{line:solution}}
		Backtrack 
	}
	\Else{
		Decide on some unassigned variable\; 
	}
}
\end{algorithm}

Let us now focus on the \codecall{lookahead} procedure. 
Before diving into the algorithmic aspects, we formalize what it aims to compute, namely a set of \emph{weighted local cores}. 
\begin{definition}
	Let $\assignment$ be an assignment, and $(\formula,\objective)$ a \maxsat instance. 
	A \emph{weighted local core} of $(\formula,\objective)$ relative to $\assignment$ is a triple 
	\[ \weightedLocalCore = \coretriple\]
	such that $\reason\subseteq \alpha$ and $\core$ contains only \emph{negations of} objective literals and
	\[F\land \reason \land \core \models \false,\]
	where $\false$ denotes the trivially false constraint $0\geq 1$. 
\end{definition}
%In other words a local core guarantees that, given the assignments in $\alpha$, the underlying objective literals cannot all be false (or in other words: at least one of the underlying objective literals should incur cost). 
In other words, a local core guarantees that, given the assignments in $\reason$ (which are also called the reasons of $\core$), the objective literals in $\core$ cannot all be false, and hence at least one of the objective literals in $\core$ has to incur cost. 

Now, we are not just interested in finding a single local core, but in a compatible \emph{set} of such cores. 
\begin{definition}
	Let $\coreset$ be a set of weighted local cores of $(\formula,\objective)$  relative to $\alpha$. 
	We call $\coreset$ $\objective$-compatible, if for each objective literal $\lit$, 
	\[\sumWithLongBottomLim{\coretriple\in\coreset \land \olnot \lit\in\core} \weight \leq \weightInObj{\lit}.\] 
\end{definition}
In other words, $\coreset$ is  $\objective$-compatible if the total weight of all cores containing a literal does not exceed the weight of the underlying objective literal in the objective.  
In what follows, if $\lit$ is an objective literal, we write $\residualw(\lit, \coreset)$ for 
\[\weightInObj{\lit} - \sumWithLongBottomLim{\coretriple\in\coreset \land \olnot \lit\in\core} \weight\] and call this the \emph{residual} weight of $\lit$ with respect to $\coreset$. The total \emph{weight} of a core set, written $\weightInObj{\coreset}$ is $\sum_{\coretriple\in\coreset}\weight$. 
The following proposition now formalizes how a set of weighted local cores can be used during search. 
\begin{proposition}\label{prop:correct} 
	Let $\coreset$ be an $\objective$-compatible set of weighted local cores and assume $\assignmenttwo$ is any total assignments that refines $\assignment$ and satisfies $\formula$. 
	\begin{enumerate}
		\item If $\weightInObj{\coreset} \geq \bestval$, then $\objective\under \assignmenttwo \geq \bestval$. 
		\item If for some unassigned objective literal $\lit$, 
		\[ \residualw{(\lit,\coreset)} + \weightInObj{\coreset}  \geq \bestval,  \]
		then it holds that if $\objective\under \assignmenttwo < \bestval$, then $\assignmenttwo(\lit)=0$. 
	\end{enumerate}
\end{proposition}
%\todo{check terminology in maxcdcl paper}
The first case of this proposition is known as a \emph{soft conflict} \cite{CLLHM25SolvingweightedMaximumSatisfiabilityBranchBound}.
This means that if the total weight of the core set exceeds $\bestval$, then we are sure that all assignments that refine the current one will have an objective value that does not improve upon the best found so far. 
In this case, we know that the current node of the search tree is hopeless and we can backtrack (in fact, \maxcdcl will learn a clause explaining this soft conflict, as discussed later). 
The second case in Proposition~\ref{prop:correct} was called \emph{hardening} \cite{CLLHM25SolvingweightedMaximumSatisfiabilityBranchBound}. 
The setting captured here is that \emph{if} we were to set a literal $\lit$ to be cost-incurring, then the condition in the first item would trigger. As such, we can safely propagate $\lit$ to be non-cost-incurring. 

Pseudo-code of the lookahead procedure, and details as to how it finds those local cores can be found in Algorithm~\ref{alg:lookahead}. 
It initializes a local assignment $\assignmentlocal$ to be equal to $\alpha$,
and maintains a partially constructed core $\core$, initialized as the empty set.  
In line~\ref{line:coreset} it initializes $\coreset$ to be the set of all trivial weighted local cores (by going over all objective literals that are already cost-incurring in $\alpha$). 
In the main loop at line~\ref{line:mainloop}, it iteratively assigns an objective literal to be non cost-incurring and runs unit propagation. 
If this results in another objective literal to be propagated to be cost-incurring (line~\ref{line:propagated}) or if a conflict is found (line~\ref{line:conflictfound}) , this means we have found a local core. 
In principle taking $R$ equal to $\assignment$ and $\core$ as constructed so far could be used for a local core. 
Instead, however, line~\ref{line:improve} uses standard conflict analysis methods to further reduce this. 
The local core that is found in this way is added to $\coreset$ (line~\ref{line:add}) and $\core$ and $\assignmentlocal$ are reset (line~\ref{line:reset}).

If one of the two conditions of \cref{prop:correct}  is satisfied, a new clause is learned (based on the set of local cores and standard conflict analysis techniques) and added to $\formula$ (lines~\ref{line:startclauseadd}--\ref{line:endclauseadd}). 
%\jordic{This is not correct, $\coresetclause$ is not added to the set of clauses (we could, but it's not).
%In the hardening case, we do not use $\coresetclause\vee \olnot\lit$ to propagate hardened literals, but the result of analyzing $\coresetclause\vee \olnot\lit$. Also in hardening, better put a for loop?: for every $\lit$ such that...} 
%\bart{Hmm... I have already included an ``improve-core'' on line 12 which should be doing the conflict analysis style reasoning. Is this not enough? Do we need to go through ANOTHER round of analysis here? (added the for loop)}
Here, we use the notation
%$
\[
\coresetclause := \bigvee_{\coretriple \in \coreset\land \lit\in\reason}\olnot \lit.
\]
%$
%is used. 

\begin{algorithm}
	\caption{Overview of the  \codecall{lookahead} procedure.  }\label{alg:lookahead}
	\KwIn{Formula $\formula$, objective $\objective$, upper bound $\bestval$, assignment $\assignment$}
%	\KwOut{optimal assigment or \textbf{UNSAT}}
%	$\formula \gets$ \codecall{Preprocess}($\formula$)\;
%	$\assignment \gets \emptyset$;\quad $\assignment^*\gets \textbf{UNSAT}$; \quad $\bestval\gets +\infty$\;
  $\assignmentlocal \gets \assignment$;\quad $\coreset\gets\emptyset$\;   $K \gets \emptyset$\;
  \For{each objective literal $\lit$ s.t.\ $\assignment(\lit) = 1$\label{line:startcorediscovery}}{
  	$\coreset \gets \coreset \cup \{\langle \weightInObj{\lit}, \{\lit\}, \{\olnot\lit\}\rangle\}$\label{line:coreset} 
  }
	\While{some $\lit$ with $\residualw(\lit,\coreset) >0$  is unassigned\label{line:mainloop}}{
		$\assignmentlocal(\lit)\gets 0$;\quad $\core \gets \core \cup\{\olnot\lit\}$\;
		\codecall{unit-propagate}($\formula$, $\assignmentlocal$)\;
		\If{some $\lit'$ with $\residualw(\lit',\coreset) >0$ is propagated\label{line:propagated}}
		{ $\core \gets \core \cup\{\olnot\lit'\}$\;}
		\ElseIf{no conflict was derived\label{line:conflictfound}}{
			\textbf{continue}}
		$\reason ,\core \gets \codecall{improve-core}(\core) $\label{line:improve}\;
		$\weight \gets \min_{\olnot\ell \in K}\residualw(\ell,\coreset)$\; 
		$\coreset \gets \coreset\cup\{\coretriple\}$\label{line:add}\;
		\lIf{$\weightInObj{\coreset} \geq \bestval$}{\textbf{break}}
	 $\core\gets\emptyset$;\quad $\assignmentlocal\gets \assignment$		\label{line:reset}
	}
	%%% One reviewer asked if this should not be a forall on line 8. However, he is wrong. Can only do this for a SINGLE l
	\If{$\weightInObj{\coreset} \geq\bestval $}{ \label{line:startclauseadd}
		$\formula\gets \formula \cup \{ \codecall{analyze}( \coresetclause)\} $\label{line:addclause-softconflict}\;
	}
	\Else{\For{each $\lit$ with $\residualw(\lit,\coreset) + \weightInObj{\coreset} \geq\bestval$ }{
	   $\formula\gets \formula \cup \{ \codecall{analyze}(  \coresetclause\vee \olnot\lit) \} $ \label{line:endclauseadd} \label{line:addclause-hardening}\;
		\codecall{unit-propagate}($\formula$, $\assignment$)\;
	}}
\end{algorithm}

We now turn our attention to proof logging for this lookahead procedure. 
It consists of two important components. 
On the one hand, we need proofs for the \emph{core discovery} (lines~\ref{line:startcorediscovery}--\ref{line:reset}) and on the other hand, the proof that the clause added to $\formula$ can indeed be derived.
We now show that this can indeed be done efficiently. 
%Given a local core $\weightedLocalCore = \coretriple$, we define the clause 
%\[\weightedLocalCoreClause := \bigvee_{\lit\in\reason\cup\core}\olnot \lit.\] 

\begin{proposition}
	\label{prop:localcoreclauseRUP}
	If $\coretriple$ is a local core found by Algorithm~\ref{alg:lookahead}, then the clause 
%	$\sum_{\lit \in \reason\cup\core} \olnot \lit\geq 1$ 
%	$\weightedLocalCoreClause$
	\[\weightedLocalCoreClause := \bigvee_{\lit\in\reason\cup\core}\olnot \lit.\]
	is RUP wrt $\formula$. 
\end{proposition}
\begin{proof}[Proof sketch]
	This follows from the fact that Algorithm~\ref{alg:lookahead} finds a local core $\weightedLocalCore$ after the propagation of the negation of $\weightedLocalCoreClause$ leads to a contradiction.
\end{proof}

\begin{theorem}\label{thm:proofmain} 
	Let $\coreset$ be an $\objective$-compatible set of weighted local cores and let $|\objective|$ be the number of objective literals. 
	\begin{enumerate}
		\item If $\weightInObj{\coreset} \geq \bestval$, then there is a cutting planes derivation that derives $\coresetclause$ from the constraint $\objective\leq\bestval-1$ and the constraints $\weightedLocalCoreClause$ for every $q \in \coreset$ using at most 
		$3 |\objective| + 2 |\coreset| + 1$ steps.
		\item If for some unassigned objective literal $\lit$, 
		$ \residualw{(\lit,\coreset)} + \weightInObj{\coreset}  \geq \bestval, $
		then there is a cutting planes derivation 
		that derives $\coresetclause \lor \olnot{\lit}$
		from the constraint $\objective\leq\bestval-1$ and the constraints $\weightedLocalCoreClause$ for every $q \in \coreset$ using at most 
		$3 |\objective| + 2 |\coreset| - 2$ steps.
	\end{enumerate}
\end{theorem}

Note that since the clauses that are finally learned by \maxcdcl in lines \ref{line:addclause-softconflict} and \ref{line:addclause-hardening} of Algorithm~\ref{alg:lookahead} are derived by standard CDCL on the clauses derived in Theorem~\ref{thm:proofmain}, they can be proven by RUP.
Moreover, when line~\ref{line:returnmaxcdcl} of Algorithm~\ref{alg:maxcdcl} finds a conflict at root level, the clause $0 \geq 1$ is also provable by RUP, which proves the optimality of the last found solution due to the objective improvement rule.

We end this section with an example illustrating the cutting planes derivation of the first case of Theorem~\ref{thm:proofmain}.\footnote{A more advanced example, as well as the proof for Theorem~\ref{thm:proofmain} can be found in \onlyAppendix{the appendix.}\notInAppendix{the extended version of this paper \cite{vandesande2025certifiedbranchandboundmaxsatsolving}.}} %no space before footnote please. 
\begin{example}
	Let us consider a MaxSAT instance expressed as a PBO instance $(\formula, \objective)$, with $\objective = 3 x_1 + 5 x_2 + 5 x_3 + 6 x_4$. Assume that the best found solution so far has objective value $\bestval = 7$, that the current partial assignment assigns $y_1$ and $y_2$ to true,
	and that algorithm~\ref{alg:lookahead} has found a set of local cores $\coreset = \{\weightedLocalCore_1, \weightedLocalCore_2\}$, with 
	\begin{align*}
		\weightedLocalCore_1 &= \langle 3, \{y_1\}, \{\olnot{x_1}, \olnot{x_2}\} \rangle\text{ \quad and} \\
		\weightedLocalCore_2 &= \langle 5, \{y_2\}, \{\olnot{x_3}, \olnot{x_4}\} \rangle\text{.}
	\end{align*} 
	We can observe that this is indeed an $\objective$-compatible set and that the first case of Theorem~\ref{thm:proofmain} is applicable. Hence, we will derive the clause $\coresetclause$, which is $\olnot{y_1} + \olnot{y_2} \geq 1$.
	
	Following Proposition~\ref{prop:localcoreclauseRUP}, we can derive the clauses
	\begin{align*}
		\clause_{\weightedLocalCore_1} &:= \olnot{y_1} + x_1 + x_2 \geq 1\text{, and} \\
		\clause_{\weightedLocalCore_2} &:= \olnot{y_2} + x_3 + x_4 \geq 1
	\end{align*}
	by Reverse Unit Propagation.
	Multiplying $\clause_{\weightedLocalCore_1}$ and $\clause_{\weightedLocalCore_2}$ by their respective weight and adding them up results in
	\begin{align}
		3 \olnot{y_1} + 3 x_1 + 3 x_2 + 5 \olnot{y_2} + 5 x_3 + 5 x_4 \geq 8 \label{eq:weightedsumcoreclauses}
	\end{align}
	Now, from the solution improving constraint, we can obtain 
	\begin{align}
		3 x_1 + 3 x_2 + 5 x_3 + 5 x_4 \leq 6 \label{eq:weakenedSIC}
	\end{align}
	by the addition of literal axioms $x_4 \geq 0$ and $x_2 \geq 0$ multiplied by $2$.
	Addition of \eqref{eq:weightedsumcoreclauses} and \eqref{eq:weakenedSIC}, and simplification gives 
	\begin{align*}
		3 \olnot{y_1} + 5 \olnot{y_2} \geq 2
	\end{align*}
	Dividing this by $5$ now yields $\coresetclause$.
\end{example}

	\section{Certifying CNF Encodings based on (Multi-Valued) Decision Diagrams}\label{sec:mdd}\label{sec:MDD}
	When a solution is found (at line~\ref{line:solution} of Algorithm~\ref{alg:maxcdcl}), \maxcdcl decides heuristically whether or not to add a CNF-encoding of the solution-improving constraint  \[\objective\leq \bestval - 1.\]
For many CNF encodings of pseudo-Boolean constraints, it is well-known how to get certification \cite{msc/Vandesande23,GMNO22CertifiedCNFTranslationsPseudo-BooleanSolving,BBNOPV24CertifyingWithoutLossGeneralityReasoningSolution-Improving}. 
One notable exception is the CNF encoding based on \emph{binary decision diagrams} (BDD), which \citet{BCSV20MDD-basedSATencodingpseudo-Booleanconstraintsat-most-one}  recently generalized to so-called \emph{multi-valued decision diagrams} (MDD)  and is used by \maxcdcl.
%This section is dedicated to describing how to obtain proof logging for those encodings. 
Most of the interesting questions from a proof logging perspective already show up for BDDs, so to keep the presentation more accessible, we focus on that case and afterwards briefly discuss how this generalizes to MDDs. 
Formal details and proofs are included in the supplementary material.

A \emph{Binary Decision Diagram (BDD)} is a (node- and \mbox{edge-)labeled} graph with 	two leaves, labeled true ($\ltrue$) and false ($\lfalse$), respectively, where  each internal node is labeled with a variable and has two outgoing edges, labeled true ($\ltrue$) and false ($\lfalse$), respectively. 
We write $\child{\node}{\lfalse}$ and $\child{\node}{\ltrue}$ for the children following the edge with labels $\lfalse$ and $\ltrue$, respectively.
Each node $\node$ in a BDD represents a \emph{Boolean function}:
the true and false leaf nodes represent a tautology and contradiction respectively;
if $\node$ is a node labeled $x$ with true child $\node_\ltrue$ and false child $\node_\lfalse$, it maps any (total) assignment $\assignment$ to $\node(\assignment)$, which is defined as $\node_\ltrue(\assignment)$ if $x\in \assignment$ and as $\node_\lfalse(\assignment)$ if $\olnot x \in \assignment$.

A BDD is \emph{ordered} if there is a total order of the variables such that each path through the BDD respects this order and it is \emph{reduced} if two conditions hold: 
\emph{(1)} no node has two identical children
and \emph{(2)} no two nodes have the same label, $\ltrue$-child and $\lfalse$-child. 
For a fixed variable ordering, each Boolean function has a unique ordered and reduced representation as a BDD, i.e., ordered and reduced BDDs form a canonical representation of Boolean functions.

When using BDDs to encode a solution-improving constraint $\obj \leq \bestval-1$ as clauses, first a BDD representing this constraint is constructed, then a set of clauses is generated from this BDD. 

\paragraph*{Phase 1: Construction of a BDD} 
The standard way to create a reduced and ordered BDD for 
\[\sum_{i=1}^n \cost_i \blocklit_i \leq \bestval-1\]
 is to first recursively descend and create BDDs $\node_\ltrue$ and $\node_\lfalse$ for 
 \begin{align*}
 	&\sum_{i>1}^n \cost_i \blocklit_i \leq v^*-1-a_1 \text{\quad and}\\ 
 	&\sum_{i>1}^n \cost_i \blocklit_i \leq v^*-1\end{align*}
 	respectively, and then combine them into the node $\bddOf{x_1}{\node_\ltrue}{\node_\lfalse}$. 
Using memoization, the BDD is always kept reduced: if a node is being created with the same label, and the same two children as an existing node, the existing node is returned instead; if $\node_\ltrue = \node_\lfalse$, then $\node_\ltrue$ is returned. 

This procedure, which first creates the two children, will however always take exponential time (in terms of the number of objective literals), while this can in many cases be avoided. 
Crucially, all PB constraints for which a BDD node are created are of the form 
$ \sum_{i\geq k}^n \cost_i \blocklit_i \leq A$ and for a fixed $k$, the fact that we work with Boolean variables means that many different right-hand-sides will result in an equivalent constraint. 
Moreover, such values can be computed while constructing the BDD using a dynamic programming approach. 
In the following proposition, which formalizes the construction, we say that $[l,u]$ is a \emph{degree interval} for $\sum_{i\geq k}^n\cost_i\blocklit_i$ when we mean that for all values $A$ in the (possibly unbounded) interval $[l,u]$, $\sum_{i\geq k}^n\cost_i\blocklit_i\leq A$ is equivalent to $\sum_{i\geq k}^n\cost_i\blocklit_i\leq u$. 
\begin{proposition}[\cite{ANORM12NewLookBDDsPseudo-BooleanConstraints}]
	If $[l_\ltrue,u_\ltrue]$ and $[\l_\lfalse,u_\lfalse]$ are degree intervals for $\sum_{i = k+1}^n\cost_i\blocklit_i$, then 
	\[ [\max(l_\ltrue+\cost_k,l_\lfalse), \min(u_\ltrue +\cost_k,u_\lfalse)]
	\]
	is a degree interval for  $\sum_{i = k}^n\cost_i\blocklit_i$.
\end{proposition}
The dynamic programming approach for creating ordered and reduced BDDs now consists in keeping track of this interval for each translated PB constraint and reusing already created BDDs whenever possible. 
From now on, we will identify a node in a BDD $\node = \bddOf{k}{l}{u}$ as the node that represents the boolean function $\sum_{i=k}^n \cost_i \blocklit_i \leq [l,u]$. 
 
\paragraph*{Phase 2: A CNF Encoding From the BDD} 
Given a BDD that represents a PB constraint, constructed as discussed above, we can get a CNF encoding as follows. 
\begin{itemize}
	\item For each internal node $\node$ in the BDD, a new variable $v_\node$ is created; intuitively this variable is true only when the Boolean function represented by $\node$ is true. In practice for the two leaf nodes no variable is created but their truth value is filled in directly. However, in the proofs below we will, to avoid case splitting pretend that a variable exists for each node. 
	%	\jordi{The truth values of the leaf nodes is not represented with new variables but with constants 1 and 0 respectively}. 
	\item For each internal node $\node$ that is labeled with literal $\blocklit$ with children $\node_\ltrue = \child{\node}{\ltrue}$ and $\node_\lfalse = \child{\node}{\lfalse}$, the clauses 
	\begin{align} \olnot \blocklit +  v_{\node_\ltrue} + \olnot v_\node & \geq 1 \text{, \quad and }\label{eq:bdd:clauseone}\\
	v_{\node_\lfalse} + \olnot v_\node & \geq 1		\label{eq:bdd:clausetwo}
	\end{align} are added.   The first clause expresses that if $\blocklit$ is true and $\node_\ltrue$ is false, then so is $\node$. 
	The second clause expresses that whenever $\node_\lfalse$ is false, so is $\node$ (this does not hold for BDDs in general, but only for the specific ones generated here, where we know $\node$ to represent a constraint \mbox{$\sum_{i\geq k}^n\cost_i\blocklit_i\leq A$} and $\node_\lfalse$ represents $\sum_{i\geq k+1}^n\cost_i\blocklit_i\leq A$, both with positive coefficients). 
	\item Finally, for the top node $\node_\top$ representing $\obj \leq \bestval-1$, the unit clause $v_{\node_\top}$ is added. 
\end{itemize}

\paragraph{Certification}
Our strategy for proof logging for this encoding follows the general pattern described by \citet{VDB22QMaxSATpbCertifiedMaxSATSolver}, namely, we first introduce the fresh variables by explicitly stating what their meaning is in terms of pseudo-Boolean constraints and later we derive the clauses from those constraints. 
Due to the fact that BDDs are reduced, the first step becomes non-obvious. 
Indeed, the new variables represent multiple (equivalent) pseudo-Boolean constraints at once. 
To make this work in practice, for each node $\node = \bddOf{k}{l}{u}$, we will show we can derive the constraints that express

\begin{align*}
%\label{eq:vnodeimplies}
	v_\node &\limplies \sum_{i\geq k}^n \cost_i \blocklit_i \leq l \text{\quad and}\\
%	, \qquad\text{which is }\quad    
%&\left(\sum_{i\geq k}^n a_i-l\right)\cdot \olnot v_\node+\sum_{i\geq k}^n -\cost_i \blocklit_i &\geq - l,\text{\qquad and}
%\\
%
%
%
%\begin{equation}
	v_\node &\limpliedby \sum_{i\geq k}^n \cost_i \blocklit_i \leq u.
%	, \qquad \text{which is }\quad 
%(u+1)\cdot v_\node+ \sum_{i\geq k}^n \cost_i \blocklit_i &\geq u+1.
%\label{eq:vnodeimplied}
\end{align*}
We will refer to these two constraints as the \emph{defining constraints} for node  $\node$. 
%Here, the first constraints states that if $v_\node$ is true, then $\sum_{i\geq k} \cost_i\blocklit_i$ is smaller than or equal to $l$ and the second constraint guarantees that if $\sum_{i\geq k} \cost_i\blocklit_i$ is smaller than or equal to $u$, then $v_\node$ must be true. 
These constraints together precisely determine the meaning of $v_\node$, but they also contain the information that the partial sum cannot take any values between $l+1$ and $u$. 
Hence, proving them is expected to require substantial effort. 
Once we have these constraints for each node, deriving the clauses becomes very easy. 
\begin{proposition}
	\label{prop:clauses-mdd}
	For each internal node $\node$, the clauses \eqref{eq:bdd:clauseone} and \eqref{eq:bdd:clausetwo} can be derived by a cutting planes derivation consisting of one literal axiom, one multiplication, three additions and one division 
	from the defining constraints for $\node$, $\child{\node}{\ltrue}$ and $\child{\node}{\lfalse}$.
\end{proposition}

\begin{proposition}
	Let $\node_\top = \bddOf{1}{l}{\bestval-1}$ for some $l \leq \bestval-1$ be the root node of the BDD. The unit clause $v_{\node_\top} \geq 1$ can be derived by a cutting planes derivation of one addition and one deletion from the defining constraints for $\node_\top$ and the solution-improving constraint. 
\end{proposition}

The harder part is to actually derive the defining constraints for the internal nodes. 
Also this can be done in the \veripb proof system, as we show next. 

\begin{proposition}\label{prop:mainproofbdd}
	Let $\node = \bddOf{k}{l}{u}$ be an internal node with label $\blocklit_k$ and with children $\node_\ltrue = \child{\node}{\ltrue}$ and $\node_\lfalse = \child{\node}{\lfalse}$.
%	 with $l = \max(l_\ltrue + \cost_k, l_\lfalse)$ and $u = \min(u_\ltrue + \cost_k, u_\lfalse)$.
	
	If the defining constraints for $\node_\ltrue$ and $\node_\lfalse$ are given, then the defining constraints for $\node$ can be derived in the VeriPB proof system.
\end{proposition}
\begin{proof}[Proof Sketch]
	The first step is to introduce two reification variables $v_\node$ and $v_\node'$ as 
	\begin{align*}
	v_\node &\lequiv \sum_{i\geq k}^n \cost_i \blocklit_i \leq l\text{\quad and} \\
	v_\node' &\lequiv \sum_{i\geq k}^n \cost_i \blocklit_i \leq u.
	\end{align*}
	The next step is showing that these two variables are actually the same; the difficult direction is showing that $v_\node'$ implies $v_\node$.  
	We do this by case splitting on $\blocklit_k$ and deriving each case by contradiction using the defining constraints for $\node_\ltrue$ and $\node_\lfalse$.
	To derive $\olnot{x} + \olnot v_\node' + v_\node \geq 1$, assume that $\blocklit_k$ and $v_\node'$ are true and that $v_\node$ is false.
	The assumption on $\blocklit_k$ and $v_\node'$ together with the defining constraints for $\node_\ltrue$ makes that $v_{\node_\ltrue}$ is true; the assumption that $\blocklit_k$ is true and $v_\node$ is false together with the defining constraints for $\node_\ltrue$ makes that $v_{\node_\ltrue}$ is false. This is clearly a contradiction.
	Using the defining constraints of $\node_\lfalse$, we can in a similar way derive $\blocklit_k + \olnot v_\node' + v_\node \geq 1$.
	A short cutting planes derivation now allows us to derive that $v_\node'$ implies $v_\node$ and hence derive the required defining constraints. 
\end{proof}

Finally, it can happen that due to reducedness of BDDs, a node actually represents multiple constraints. 
The following proposition shows that also this can be proven in VeriPB.
\begin{proposition}
	Consider a node $\node = \bddOf{k}{l}{u}$ and let $k'$ and $l'$ be numbers with $k' < k$ and $l' = l + \sum_{i=k'}^{k-1} \cost_i$ such that the Boolean functions $\sum_{i=k}^n \cost_i \blocklit_i \leq u$ and $\sum_{i=k'}^n \cost_i \blocklit_i \leq u$ are equivalent.
	
	From the defining constraints of $\bddOf{k}{l}{u}$, there is a cutting planes derivation consisting of 
	$6(k - k') + 3$ steps
%	$2(k-k') + 1$ literal axioms, each followed by a multiplication and addition
	that derives the defining constraints for $\bddOf{k'}{l'}{u}$ for reification variable $v_\node$.
\end{proposition}

\paragraph{Generalization to MDDs}

\maxcdcl not only makes use of BDDs for encoding the solution-improving constraint, but also of MDDs. 
The idea is as follows.
In some cases, \maxcdcl can infer implicit at-most-one constraints. 
These are constraints of the form $\sum_{i \in I } \blocklit_i \leq 1$ where the $\blocklit_i$ are literals in the objective $\objective$. 
The detection of such constraints is common in \maxsat solvers, and certification for it has been described by \citet{BBNOV23CertifiedCore-GuidedMaxSATSolving}, so we will not repeat this here. 
%% Reviewer suggested to remove "so we will not repeat this here", I guess this is kept on purpose.
Now assume that a set of disjoint at-most-one constraints has been found. 
In an MDD, instead of branching on \emph{single variable} in each node, we will branch on a set of variables for which an at-most-one constraint has previous been derived. 
This means a node does not have two, but $|I|+1$ children: one for each variable in the set and one for the case where none of them is true. 
Otherwise, the construction and ideas remain the same. 
As far as \emph{certification} is concerned: essentially all proofs continue to hold; the main difference is that the case splitting in the proof of \cref{prop:mainproofbdd} will now split into $|I|+1$ cases instead of two and will then use the at most one constraint to derive that the conclusion must hold in exactly one of the cases.

	\section{Experiments}\label{sec:experiments}
	%\todo{Maybe at the start of the experiments mention explicitly that a lot of techniques that are not discussed in this paper are also equipped with proofs, this includes , subsumption, or equivalent literal detection, ...}

We implemented proof logging in the \maxcdcl solver as described by \citet{COLL2025107195}.\footnote{On top of the basic algorithm described above, the experiments include proof logging for a range of techniques included in MaxCDCL \cite{COLL2025107195}, namely clause vivification, equivalent and failed literal detection, clause subsumption and simplification due to unit propagation.}
Each solver call was assigned a single core on a 2x32-core AMD EPYC 9384X. 
The time and memory limits are enforced by runlim.\footnote{https://github.com/arminbiere/runlim} %no space before footnote please 
As benchmarks, we used the instances of the MaxSAT Evaluation 2024~\cite{url:MSE24}.
The implementation, together with the raw data from the experiments can be found on Zenodo~\cite{VCB26Experimental}.

\paragraph{Overhead in proof logging}
A scatter plot evaluating the overhead in the solving time due to proof logging is shown in Figure~\ref{fig:maxcdcl-mdd8withoutvswithPL}.
The experiments ran with a time limit of 1 hour and a memory limit of 32GB.
Out of the $701$ instances that were solved without proof logging, six instances could not be solved with proof logging. 
Considering instances solved both with and without proof logging, the median overhead of proof logging is $19\%$, which shows that for most instances, proof logging overhead is manageable.
However, we see that for $10\%$ of the instances, \maxcdcl with proof logging takes more than $4.61$ times the time to solve the instance using \maxcdcl without proof logging.
This is in line with earlier work on certifying PB-to-CNF encodings using VeriPB, such as the work of \citet{BBNOPV24CertifyingWithoutLossGeneralityReasoningSolution-Improving}.
Preliminary tests indicate as possible explanation for this phenomenon that writing the defining constraints for a node in an MDD is linear in the number of objective literals.
In practice, this leads to large overhead for instances with a large number of objective literals.
Investigating possible ways to circumvent this is part of planned future work.
\begin{figure}
	\centering
	\includegraphics[scale=0.6]{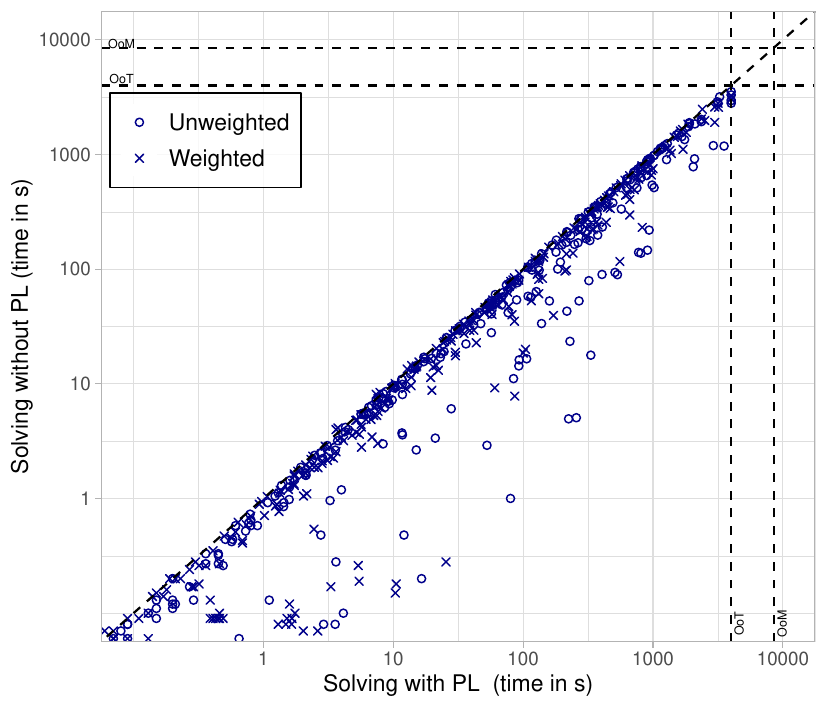}
	\caption{Comparison of solving time with and without proof logging.}
	\label{fig:maxcdcl-mdd8withoutvswithPL}
\end{figure}

\paragraph{Overhead in proof checking}

While optimizing the checking time for the produced proofs was out of scope for this work, we also evaluated it by running the VeriPB proof checker~\cite{url:veripb} with a time limit of 10 hours and a memory limit of 64GB.
Figure~\ref{fig:maxcdcl-mdd8solvingvschecking} shows a scatter plot comparing the solving time with proof logging with the proof checking time.
Out of $695$ produced proofs, $485$ were checked successfully; for the other $210$ instances, the proof checker ran out of time for $202$ instances and out of memory for $8$ instances. In the median proof verification took $42.94$ times the solving time. 
We believe this checking overhead to be too high for making this usable in practice. 
However, there are several ways in which this can be improved.
On the one hand, a new proof checker is currently under development for improved performance~\cite{url:pboxide}. 
Simultaneously, there are ongoing efforts for building a proof \emph{trimmer} for VeriPB proofs. 
The proofs generated by our tools can serve as benchmarks to further guide this ecosystem of VeriPB tools. 
On the other hand, once the proof checker is in a stable state, we can investigate which potential optimizations \emph{to the generated proofs} benefit the proof checking time. 
One thing that comes to mind is using \emph{RUP with hints} instead of plain RUP, so that the checker does not need to perform unit propagation.\footnote{For more information on these rules, see the VeriPB repository \cite{url:veripb}.}
\begin{figure}
	\centering
	\includegraphics[scale=0.6]{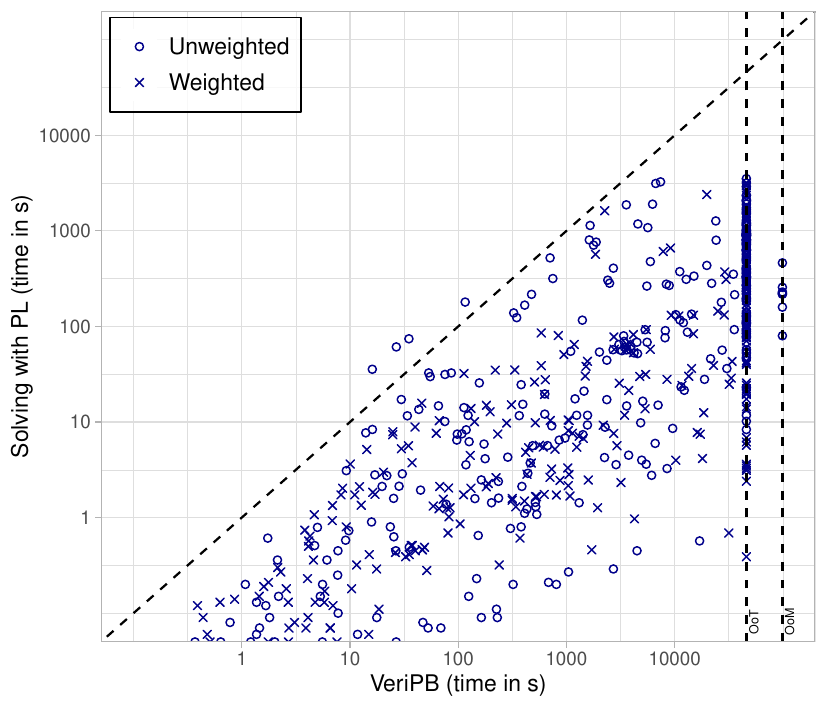}
	\caption{Comparison of solving time (with proof logging enabled) with proof checking time.}
	\label{fig:maxcdcl-mdd8solvingvschecking}
\end{figure}

%Some things that come to mind are \emph{(1)} using RUP with hints instead of plain RUP so that the checker does not need to perform unit propagation, and \emph{(2)} deleting intermediate constraints used in the the MDD encoding to optimize memory usage of the checker.\footnote{For more information on these rules, see the VeriPB repository \cite{url:veripb}.} %No space before footnote please
%Moreover, the produced proofs can serve as benchmarks for the new proof checker currently under development for improved performance~\cite{url:pboxide}. 
%\bart{Sounds weird to mention RUP with hints without actually implementing it. Why not? (one reviewer asked) }
%\dieter{I meant with this part that we did not yet focus on improving checking time, but we could continue working on MaxCDCL to also improve on the checking part as future work. But now that I read it back, it would be better to (at least also) stress that performance enhancement is needed in the checker engineering.}

%\bart{One reviewer said (rightly): deeper analysis of experiemnts would be needed (where does the overhead come from?)}
%
%\bart{Also run unweighted instances?}
	
	\section{Conclusion}\label{sec:concl}
	In this paper, we demonstrate how to add proof logging to the branch-and-bound MaxSAT solving paradigm using the VeriPB proof system.
To do so, we formalize the paradigm with proof logging in mind, offering a new perspective on how these solvers operate.
To put our formalization to the test, we implemented proof logging in the state-of-the-art branch-and-bound solver \maxcdcl, including the MDD encoding of the solution-improving search constraint, which is created by \maxcdcl to improve its performance.
The main challenge in certifying the MDD encoding turns out to be proving that a node in an MDD represents multiple equivalent Boolean functions.
Experiments show that proof logging overhead is generally manageable, while proof checking performance still leaves room for improvement.

As future work, we plan to add proof logging to literal unlocking, a recent technique \cite{aaai/Li0CHM25} that improves lower bounds during look-ahead. 

% and to the so-called implicit hitting set solvers~\cite{DB11SolvingMAXSATSolvingSequenceSimplerSAT, DB13ExploitingPowermipSolversMaxSAT, DB13PostponingOptimizationSpeedUpMAXSATSolving, SBJ16LMHSSAT-IPHybridMaxSATSolver}, the only remaining MaxSAT paradigm without proof logging support.
%THIS IS DONE!

	\section*{Acknowledgments}
	This work is partially funded by the European Union (ERC, CertiFOX, 101122653). Views and opinions expressed are however those of the author(s) only and do not necessarily reflect those of the European Union or the European Research Council. Neither the European Union nor the granting authority can be held responsible for them.

In addition, this work is partially funded by the Fonds Wetenschappelijk Onderzoek -- Vlaanderen (projects G064925N and G070521N \& fellowship 11A5J25N),
and by grants PID2021-122274OB-I00 and PID2024-157625OB-I00 funded by MICIU/AEI/10.13039/501100011033/FEDER, UE.

	\appendix
	
	\section{Preliminaries}
	We start this appendix with a slightly more in-depth discussion of the redundance-based strengthening rule. For an even more detailed description, the reader is refered to the work by \citeauthor{BGMN22CertifiedSymmetryDominanceBreakingCombinatorialOptimisation}~(\citeyear{BGMN22CertifiedSymmetryDominanceBreakingCombinatorialOptimisation}).
\veripb allows deriving non-implied constraints with the \emph{redundance-based strengthening rule}:
a generalization of the RAT rule~\cite{JHB12InprocessingRules} and commonly used in proof systems for SAT. 
{
	This rule makes use of a \emph{substitution}, which maps every variable to $0$, $1$, or a literal.
	Applying a substitution $\sub$ on a constraint $C$ results in the constraint ${C}\under{\sub}$	obtained from $C$ by replacing each $x$ by $\sub(x)$.
	\begin{description}
		\item [Redundance-based strengthening:] If $\proofconfig \land \lneg{C}\models \objective\leq \subst{\objective}{\sub} \land \subst{(\proofconfig \land C)}{\sub}$, we add $C$ to $\proofconfig$.
	\end{description}
	Intuitively, this rule can be used to show that $\sub$, when applied to assignments instead of formulas, maps any solution of $\proofconfig$ that does not satisfy $C$ to a solution of $\proofconfig$ that
	does satisfy $C$ and that has an objective value that is at least as good.
}
In this paper, the redundance rule is mainly used for \textbf{reification}: for any PB constraint $C$ and any fresh variable $v$, not used before, two applications of the redundance rule allow us to derive PB constraints that express $v\limplies C$ and $v\limpliedby C$. In other words, when we derive a constraint $\constraint$ from the proof configuration $\proofconfig$ by redundance-based strengthening, we have the guarantee that $\proofconfig$ and $\proofconfig \land \constraint$ are equi-optimal.

In the rest of this paper, we only use reification for constraints $C$ of the form $\sum_i a_i x_i \leq A$. The \emph{PB constraints representing the reification} are denoted as 
\begin{align}
	&\Creif{\limplies}{v} := ((\sum_i a_i) - A) \olnot v - \sum_i a_i x_i \geq -A \text{, and} \label{eq:def:rightreif} \\
	&\Creif{\limpliedby}{v} := (A+1) v + \sum_i a_i x_i \geq A+1 \label{eq:def:leftreif}.
\end{align}
Observe that \eqref{eq:def:rightreif} indeed coincides with $v \limplies C$: if $v$ is true, then $\Creif{\limplies}{v}$ is equivalent with $C$; if $v$ is false, $\Creif{\limplies}{v}$ is trivially true.
Similarly, one can convince themselves that $\Creif{\limpliedby}{v}$ coincides with $v \limpliedby C$ by observing that $\Creif{\limpliedby}{v}$ actually means $\olnot v \limplies \sum_i a_i x_i \geq A+1$.
Indeed, since $v$ was a fresh variable that was not used before, for every solution to $\proofconfig$ we can assign a value to $v$ such that  $\proofconfig \land \Creif{\limpliedby}{v} \land \Creif{\limplies}{v}$ is satisfied. 

To end this section, we will show that the redundance rule also allows us to derive implied constraints by a \textbf{proof by contradiction}, namely that if we have a cutting planes derivation that shows that $\proofconfig \land \lneg \constraint \models 0 \geq 1$, then we can add $\constraint$ to $\proofconfig$. {In practice, this is mimicked by applying the \emph{redundance-based strengthening rule} with an empty witness.}
{
	To see that this is indeed correct, consider the definition of the redundance-based strengthening proof rule: if $\sub$ is the empty witness, then we need to prove that $\proofconfig \land \lneg \constraint  \models \constraint$, which in practice boils down to showing that from the constraints $\proofconfig \land \lneg \constraint \land \lneg \constraint$, a contradiction can be derived by a Cutting Planes derivation.
	In what follows, we will say that we prove a constraint by contradiction, while in practice, we prove this constraint by redundance-based strengthening with an empty witness.
}

	\section{Certifying Branch-And-Bound with Look-Ahead} \label{sec:maxcdcl-appendix}
	We start by giving the proof for Theorem~\ref{thm:proofmain}

\setcounternextresult{thm:proofmain}
\begin{theorem}
	Let $\coreset$ be an $\objective$-compatible set of weighted local cores and let $|\objective|$ be the number of objective literals. 
	\begin{enumerate}
		\item If $\weightInObj{\coreset} \geq \bestval$, then there is a cutting planes derivation that derives $\coresetclause$ from the constraint $\objective\leq\bestval-1$ and the constraints $\weightedLocalCoreClause$ for every $q \in \coreset$ using at most 
%		$3 |\objective| + 2 |\coreset| + 2$ steps.
		$|\objective|$ literal axioms, $|\objective| + |\coreset|$ multiplications, $|\objective| + |\coreset|$ additions, and 1 division.
		\item If for some unassigned objective literal $\lit$, 
		$ \residualw{(\lit,\coreset)} + \weightInObj{\coreset}  \geq \bestval, $
		then there is a cutting planes derivation 
		that derives $\coresetclause \lor \olnot{\lit}$
		from the constraint $\objective\leq\bestval-1$ and the constraints $\weightedLocalCoreClause$ for every $q \in \coreset$ using at most 
%		$3 |\objective| + 2 |\coreset| + 2$ steps.
		$|\objective|-1$ literal axioms, $|\objective| + |\coreset|-1$ multiplications, $|\objective| + |\coreset| - 1$ additions, and 1 division.
	\end{enumerate}
\end{theorem}

\begin{proof}
	We start by showing the first case. Multiplying the clauses $\constraint_{\coretriple}$ for every $\coretriple \in \coreset$ by their weight $w$ and adding them up gives 
	\begin{align}
		\sum_{\coretriple \in \coreset} \left(\sum_{\lit \in R} w \olnot{\lit} + \sum_{\olnot{\lit} \in K} w {\lit}  \right) \geq \weightInObj{\coreset}
	\end{align}
	which is equivalent to 
	\begin{align}
	\label{eq:sumlocalcores}
	\sum_{\coretriple \in \coreset, \lit \in R} w \olnot{\lit} + \sum_{\coretriple \in \coreset, \olnot{\lit} \in K} w {\lit} \geq \weightInObj{\coreset}
	\end{align}
	By adding literal axioms ${\lit} \geq 0$ for every objective literal $\lit$, multiplied by $\residualw({\lit}, \coreset)$ to the solution-improving constraint $\objective \leq \bestval-1$ we can get
	\begin{align}
		\label{eq:weakenedobjectiveimprovingconstraint}
		\sum_{\coretriple \in \weightedLocalCore, \olnot{\lit} \in K}  w {l} \leq \bestval - 1
	\end{align}
	Addition of \eqref{eq:sumlocalcores} and \eqref{eq:weakenedobjectiveimprovingconstraint} gives
	\begin{align}
		\sum_{\coretriple \in \weightedLocalCore, l \in R} w \olnot{l} \geq \weightInObj{\coreset} - \bestval + 1
	\end{align}
	Since the first case applies, we know that $\weightInObj{\coreset} \geq \bestval$, so the right-hand side is strictly positive. Division by a large enough number (and rounding up) results in the clause $\weightedLocalCoreClause$.
	
	\newcommand{\litp}{\m{\lit'}}
	Let us now assume the second case applies for some objective literal $\litp$, and thus we need to derive $\coresetclause \lor \olnot{\litp}$. We start by adding literal axioms ${\lit} \geq 0$ multiplied by $\residualw({l}, \coreset)$ for every objective literal ${\lit} \neq \litp$ to the solution-improving constraint $\objective \leq \bestval-1$ to get
	\begin{align}
		\label{eq:weakenedobjectiveimprovingconstraint2}
		\residualw(\litp, \coreset) \litp + \sum_{\coretriple \in \coreset, \olnot{\lit} \in K} w \lit  \leq \bestval - 1 
	\end{align}
	By addition of \eqref{eq:weakenedobjectiveimprovingconstraint2} with \eqref{eq:sumlocalcores}, which we can derive by taking the same linear combination of the local cores in $\coreset$, we get
	\begin{align}
		\sum_{\coretriple \in \coreset, l \in R} w \olnot{\lit} - \residualw(\litp, \coreset) \litp \geq \weightInObj{\coreset} - \bestval + 1
	\end{align}
	which is, given that $\litp = 1 - \olnot{\litp}$, equivalent to
	\begin{align}
		\sum_{\coretriple \in \coreset, l \in R} w \olnot{\lit} &+ \residualw(\litp, \coreset) \olnot{\litp}  \nonumber \\ & \geq \weightInObj{\coreset} - \bestval + 1 + \residualw(\litp, \coreset)
	\end{align}
	Now, since the second case applies for $\litp$, we know that $\weightInObj{\coreset} + \residualw(\coreset, \litp) \geq \bestval$, so the right-hand side is positive. Hence, division by a large enough number (and rounding up) gives $\weightedLocalCoreClause \lor \olnot{\litp}$.
\end{proof}

Next, we give a more in-depth example of the look-ahead procedure, as well as how the proof logging works.

\begin{example}
	Let us consider a MaxSAT instance with objective $\objective = 5 x_1 + 9 x_2 + 6 x_3 + 2 x_4 + 6 x_5$. The current best found solution has objective value $\bestval = 12$ and the current partial assignment when the solver starts Algorithm~\ref{alg:lookahead} is $\{x_1, \olnot{y_1}, y_2, y_3\}$. We start this example with a trace of Algorithm~\ref{alg:lookahead}.

	\textbf{Line 3-4} Add $q_0 = \langle 5, \{x_1\}, \{\olnot{x_1}\} \rangle$ to $\coreset$.
	
	\textbf{Line 5-7} Take some literal $l$ with $\residualw(l, \coreset) > 0$. Take $x_2$. Propagate $\olnot{x_2}$. 
	
	\textbf{Line 8-11} No conflict propagated, so search for local core continues.
	
	\textbf{Line 5-7} Take next literal $l$ with $\residualw(l, \coreset) > 0$. Take $x_3$. Propagate $\olnot{x_3}$. 
	
	\textbf{Line 8-11} A conflict was propagated. Found local core $K = \{\olnot{x_2}, \olnot{x_3}\}$.
	
	\textbf{Line 12-13} Conflict analysis: assignment $\{\olnot{y_1}, y_3, \olnot{x_2}, \olnot{x_3}\}$ propagates to conflict. $R = \{\olnot{y_1}, y_3\} , K = \{x_2, x_3\}$, $w_1 = \min\{\residualw(x_2, \coreset), \residualw(x_3, \coreset)\} = 6$. 
	
	\textbf{Line 14} Add core $q_1 = \langle 6, \{\olnot{y_1}, y_3\}, \{\olnot{x_2}, \olnot{x_3}\}$ to $\coreset$. Residual weight update: $\residualw(x_2, \coreset) = 3, \residualw(x_3, \coreset) = 0$. 
	
	\textbf{Line 15} $\sum_{\coretriple \in \coreset} w < \bestval$, continue search for local cores.
	
	\textbf{Line 5-7} Take some literal $l$ with $\residualw(l, \coreset) > 0$. Take $x_2$. Propagate $\olnot{x_2}$. 
	
	\textbf{Line 8-11} No conflict propagated, so search for local core continues. 
	
	\textbf{Line 5-7} Take some literal $l$ with $\residualw(l, \coreset) > 0$. Take $x_4$. Propagate $\olnot{x_4}$.  
	
	\textbf{Line 8-11} $x_5$ propagated true. Found local core $K = \{\olnot{x_2}, \olnot{x_4}, \olnot{x_5}\}$.
	
	\textbf{Line 12-13} Conflict analysis: assignment $\{\olnot{y_1}, y_2, \olnot{x_4}, \olnot{x_5}\}$ propagates to conflict. $R = \{\olnot{y_1}, y_2\}, K = \{\olnot{x_4}, \olnot{x_5}\}$, $w_2 = \min\{\residualw(x_4, \coreset), \residualw(x_5, \coreset)\} = 2$.
	
	\textbf{Line 14} Add core $q_2 = \langle 2, \{\olnot{y_1}, y_2\}, \{\olnot{x_4}, \olnot{x_5}\}$ to $\coreset$. 
	
	\textbf{Line 15} $\sum_{\coretriple \in \coreset} w \geq \bestval$, soft conflict found.
	
	\textbf{Line 18} Introduce clause $\coresetclause = \olnot{x_1} \lor y_1 \lor \olnot{y_3} \lor y_2$. Result $\codecall{analyze}( \coresetclause)$ is $\olnot{x_1} \lor \olnot{y_2}$, which is learned by the solver.
	
	\paragraph{Proof logging}
	We will now show how to derive $\coresetclause$ by a cutting planes derivation. After that, the learned clause can be derived by RUP, since it is found by standard CDCL from $\coresetclause$.
	
	First, we derive by RUP the clauses $\clause_\weightedLocalCore$ for every core $q \in \coreset$:
	\begin{align}
		\clause_{\weightedLocalCore_1} := y_1 + \olnot{y_3} + x_2 + x_3 \geq 1 \\
		\clause_{\weightedLocalCore_2} := y_1 + \olnot{y_2} + x_4 + x_5 \geq 1 
	\end{align}
	Addition of $\clause_{\weightedLocalCore_1}$, multiplied by $w_1$, with $\clause_{\weightedLocalCore_2}$, multiplied by $w_2$ gives
	\begin{align}
		8 y_1 + 6 \olnot{y_3} +  2 \olnot{y_2} + 6 x_2 +  6 x_3 + 2 x_4 + 2 x_5 \geq 8 \label{eq:example:sumofweightedcores}
	\end{align}
	By adding literal axioms literal axioms $x_2 \geq 0$ multiplied by $3$ and $x_5 \geq 0$ multiplied by 4 to the solution-improving constraint, we get
	\begin{align}
		5 x_1 + 6 x_2 + 6 x_3 + 2 x_4 + 2 x_5 \leq 11 \label{eq:example:weakenedsic}
	\end{align}
	Addition of \eqref{eq:example:sumofweightedcores} and \eqref{eq:example:weakenedsic} gives
	\begin{align}
		8 y_1 + 6 \olnot{y_3} +  2 \olnot{y_2} \geq -3 + 5 x_1
	\end{align}
	Since $x_1 = 1 - \olnot{x_1}$, we can rewrite this to
	\begin{align}
		5 \olnot{x_1} + 8 y_1 + 6 \olnot{y_3} +  2 \olnot{y_2} \geq 2
	\end{align}
	Division by $8$ and rounding up gives indeed $\coresetclause$.
\end{example}

	\section{Certifying CNF Encodings based on Binary (and Multi-Valued) Decision Diagrams}\label{sec:mdd-appendix}\label{sec:MDD-appendix}
	
\subsection{Generalizing BDD encodings to MDD encodings}
\emph{Multi-valued Decision Diagrams} (or MDD in short) are a generalization of the Binary Decision Diagrams that were discussed in Section~\ref{sec:MDD}.
Let us assume that we will encode the solution-improving constraint $\sum_{i=1}^{n} \cost_i \blocklit_i < \bestval$ and that we have a partition of the integers from $1$ to $n$, denoted as $\langle \litspartition_1, \cdots, \litspartition_m \rangle$, such that for every $\litspartition_j$, we have the constraint $\amo(\litspartition_j) = \sum_{i \in \litspartition_j} \blocklit_i \leq 1$.
We will call a \emph{layer} an index $j$ between $1$ and $m$. % that points to a set of literals $\litspartition_j$ in the partition.
A node $\node = \mddOf{k}{l}{u}$ for some layer $k$ and degree interval $[l,u]$ represents the boolean function $\sum_{j = k}^{m} \sum_{i \in \litspartition_j} \cost_i \blocklit_i \leq [l,u]$ and has $|\litspartition_j| + 1$ children: one for each literal $\blocklit_i$ with $i \in \litspartition_j$ representing when $\blocklit_i$ is true, and an else node $\node_\lfalse$, for when all literals in $\litspartition_j$ are false.
For every node $\node = \mddOf{k}{l}{u}$ and literal $\blocklit_i \in \litspartition_k$, we write $\child{\node}{\blocklit_i}$ for the child of $\node$ with an edge labeled $\blocklit_i$ and $\child{\node}{\lfalse}$ for the else-child.

\paragraph*{Phase 1: Construction of an MDD for $\obj \leq \bestval-1$} 
The creation of the MDD is formalized in Proposition~\ref{prop:mddcreation}. Observe that since this is a generalization of the creation of BDDs, as it was discussed in Section~\ref{sec:mdd}, the same dynamic programming approach can be followed to create the MDD.

\begin{proposition}[\cite{BCSV20MDD-basedSATencodingpseudo-Booleanconstraintsat-most-one, ANORM12NewLookBDDsPseudo-BooleanConstraints}]
	\label{prop:mddcreation}
	If $k$ is a layer and  $\{[l_i, u_i] \mid i \in \litspartition_{k}\} \cup \{[l_\lfalse, u_\lfalse]\}$ is a set of $|\litspartition_{k}|+1$ degree intervals for $\sum_{j \geq k+1} \sum_{i \in \litspartition_j} \cost_i \blocklit_i$, then 
	\[
		[\max(\{l_i + \cost_i \mid i \in \litspartition_k\} \cup \{l_\lfalse\}), \min(\{u_i + \cost_i \mid i \in \litspartition_k\} \cup \{u_\lfalse\}) ]
	\]
	is a degree interval for $\sum_{j \geq k} \sum_{i \in \litspartition_j} \cost_i \blocklit_i$.
\end{proposition}

\paragraph*{Phase 2: A CNF Encoding From the MDD} 
Given an MDD that represents a PB constraint, constructed as discussed above, we can get a CNF encoding as follows. 
\begin{itemize}
%	\item For each internal node $\node$ in the BDD, a new variable $v_\node$ is created; intuitively this variable is true only when the Boolean function is true. In practice for the two leaf nodes no variable is created but their truth value is filled in directly. However, in the proofs below we will, to avoid case splitting pretend that a variable exists for each node. 
	%	\jordi{The truth values of the leaf nodes is not represented with new variables but with constants 1 and 0 respectively}. 
	\item For each internal node $\node$ at layer $j$ and child node $\node_c = \child{\node}{\blocklit_i}$ for any $i \in \litspartition_j$, the clauses
	\begin{align} \olnot \blocklit_i +  v_{\node_c} + \olnot v_\node & \geq 1 \label{eq:mdd:clauseone}
%	v_{\node_\lfalse} + \olnot v_\node & \geq 1		\label{eq:mdd:clausetwo}
	\end{align} are added and for the else-child $\node_\lfalse = \child{\node}{\lfalse}$, the clause
	\begin{align}
	v_{\node_\lfalse} + \olnot v_\node & \geq 1
	\end{align}   
	is added.
	
	The first clause expresses that if $\blocklit_i$ is true and $v_{\node_c}$ is false, then so is $v_\node$. 
	The second clause expresses that whenever $v_{\node_\lfalse}$ is false, so is $v_\node$ (this does not hold for MDDs in general, but only for the specific ones generated here, where we know $\node$ to represent a constraint \mbox{$\sum_{i\geq k}^n\cost_i\blocklit_i\leq A$} and $\node_\lfalse$ represents $\sum_{i\geq k+1}^n\cost_i\blocklit_i\leq A$, both with positive coefficients). 
	\item Finally, for the top node $\node_\top$ representing $\obj \leq \bestval-1$, the unit clause $v_\top$ is added. 
\end{itemize}

\subsection{Proofs for MDD encodings}

We start by considering the at-most-one constraint $\amo(\litspartition_k)$ for some layer $k$, and derive an upper bound on the objective value for layer $k$.
\begin{aproposition}
	If $k$ is a layer, the constraint 
	$\amoUB(\litspartition_k) := \sum_{i \in \litspartition_k} \cost_i \blocklit_i \leq \max(\{\cost_i \mid i \in \litspartition_k\})$
	can be derived
	from the constraint $\amo(\litspartition_k)$ using at most 
	$|\litspartition_k|$ multiplications
 	$|\litspartition_k| - 1$ literal axioms and 
 	$|\litspartition_{k} - 1$ additions. 
\end{aproposition} 
\begin{proof}
	Multiplying $\amo(\litspartition_k)$ by $\max(\{\cost_i \mid i \in \litspartition_k\})$ and adding literal axioms $\blocklit_i \geq 0$ multiplied by $\max(\{\cost_i \mid i \in \litspartition_k\}) - \cost_i$ for every $i \in \litspartition_k$ gives the desired constraint.
\end{proof}

We now introduce the defining constraints for a node formally.
\begin{adefinition}
	Let $\node = \mddOf{k}{l}{u}$ for some layer $k$ that represents the Boolean function $\sum_{j \geq k} \sum_{i \in \litspartition_j} \cost_i \blocklit_i \leq [l,u]$. 
	The \emph{defining constraints} for node $\node$ are the constraints representing the reification
	\begin{align}
%		\Cdef{\limplies}{\node}klu := 
		v_\node \limplies \sum_{j \geq k} \sum_{i \in \litspartition_j} \cost_i \blocklit_i \leq l \label{eq:reif_node:right} \\
%		\Cdef{\limpliedby}{\node}klu := 
		v_\node \limpliedby \sum_{j \geq k} \sum_{i \in \litspartition_j} \cost_i \blocklit_i \leq u \label{eq:reif_node:left}
	\end{align}  
\end{adefinition}

The implications \eqref{eq:reif_node:right} and \eqref{eq:reif_node:left} are in pseudo-Boolean form equal to
\begin{align}
%	\Cdef{\limplies}{v_\node}{k}{l}{u} &:= 
	&((\sum_{j \geq k} \sum_{i \in \litspartition_j} \cost_i) - l) \olnot{v_\node} - \sum_{j \geq k} \sum_{i \in \litspartition_j} \cost_i \blocklit_i \geq -l \text{, and}  \\
%	\Cdef{\limpliedby}{v_\node}klu &:= 
	&(u+1) v_\node + \sum_{j \geq k} \sum_{i \in \litspartition_j} \cost_i \blocklit_i \geq u+1 \text{,}
\end{align}
which will be referred to as $\Cdef{\limplies}{v_\node}{k}{l}{u}$ and $\Cdef{\limpliedby}{v_\node}klu$,
respectively. 

%%%% -----------------------------------------------------
\paragraph{Deriving the clauses in the MDD encoding}
From these constraints, we can derive the clauses representing the MDD easily. We will show this first, before moving on to explaining how to derive the defining constraints for a node.

\setcounternextresult{prop:clauses-mdd}
\begin{proposition}
	Let $\node = \mddOf{k}{l}{u}$ be an internal node, $\node_c = \child{\node}{\blocklit_z}$ a child of $\node$ for some $z \in \litspartition_{k}$ and $\node_\lfalse = \child{\node}{\lfalse}$ the else-child for $\node$, then the clauses
	\begin{align}
		\olnot \blocklit_i + v_{\node_c} + \olnot{v_\node} \geq 1 \label{eq:mdd:clause_var} \\
		v_{\node_\lfalse} + \olnot{v_\node} \geq 1 \label{eq:mdd:clause_else}		
	\end{align} 
	can be derived by a cutting planes derivation consisting of $2 |\litspartition_{k}| + 1$ additions, $2 |\litspartition_{k}| - 1$ literal axioms, $2 |\litspartition_{k}| - 1$ multiplications and $2$ divisions.
\end{proposition}
\begin{proof}
	Let us start by deriving \eqref{eq:mdd:clause_else}. Addition of $\Cdef{\limplies}{v_\node}{k}{l}{u}$ and $\Cdef{\limpliedby}{v_{\node_\lfalse}}{k+1}{l_\lfalse}{u_\lfalse}$ gives
	\begin{align}
	((\sum_{j \geq k} &\sum_{i \in \litspartition_{k}}) - l) \olnot{v_\node} + (u_\lfalse + 1) v_{\node_\lfalse} \\ 
	& - (\sum_{j \geq k} \sum_{i \in \litspartition_j} \cost_i \blocklit_i) + (\sum_{j \geq k+1} \sum_{i \in \litspartition_j} \cost_i \blocklit_i) 
	\geq -l + u_\lfalse + 1
	\end{align}
	which is 
	\begin{align}
	((\sum_{j \geq k} \sum_{i \in \litspartition_{k}}) - l) \olnot{v_\node} + (u_\lfalse + 1) v_{\node_\lfalse} - \sum_{i \in \litspartition_k} \cost_i \blocklit_i \geq -l + u_\lfalse + 1
	\end{align}
	Adding literal axioms $\blocklit_i \geq 0$, multiplied by $\cost_i$ for every $i \in \litspartition_k$ gives
	\begin{align}
	((\sum_{j \geq k} \sum_{i \in \litspartition_{k}}) - l) \olnot{v_\node} + (u_\lfalse + 1) v_{\node_\lfalse}  \geq -l + u_\lfalse + 1
	\end{align}
	From the construction of MDDs, as explained by Proposition~\ref{prop:mddcreation}, we know that $u_\lfalse \geq u \geq l$, and therefore, by dividing by a large enough number (and rounding up), we indeed get clause \eqref{eq:mdd:clause_else}.
	
	Now, to derive \eqref{eq:mdd:clause_var}, adding up the constraints $\Cdef{\limplies}{v_\node}{k}{l}{u}$ and $\Cdef{\limpliedby}{v_{\node_c}}{k+1}{l}{u}$ gives
	\begin{align}
	((\sum_{j \geq k} \sum_{i \in \litspartition_{k}}) - l) \olnot{v_\node} + (u_c + 1) v_{\node_\lfalse} - \sum_{i \in \litspartition_k} \cost_i \blocklit_i \geq -l + u_c + 1
	\end{align}
	Addition of literal axioms $\blocklit_i \geq 0$, multiplied by $\cost_i$ for every $i \in \litspartition_j \setminus \{z\}$ gives
	 \begin{align}
	 ((\sum_{j \geq k} \sum_{i \in \litspartition_{k}}) - l) \olnot{v_\node} + (u_c + 1) v_{\node_\lfalse} - \cost_z \blocklit_z \geq -l + u_\lfalse + 1
	 \end{align}
	 Because $\blocklit_z$ is equal to $1 - \olnot{\blocklit_z}$, we can rewrite this now to
	  \begin{align}
	 ((\sum_{j \geq k} \sum_{i \in \litspartition_{k}}) - l) \olnot{v_\node} + (u_c + 1) v_{\node_\lfalse} + \cost_z \olnot{\blocklit_z} \geq -l + u_\lfalse + \cost_z + 1 
	 \end{align}
	 From Proposition~\ref{prop:mddcreation} we know that $u_c + \cost_z \geq u \geq l$, and therefore we can divide by a large enough number to get indeed constraint \eqref{eq:mdd:clause_var}.
\end{proof}

\begin{proposition}
	Let $\node_\top = \bddOf{1}{l}{\bestval-1}$ be the root node of the BDD. The clause $v_{\node_\top} \geq 1$ can be derived by a cutting planes derivation of one addition and one division from the defining constraints for $\node_\top$ and the solution improving constraint. 
\end{proposition}

\begin{proof}
	Adddition of the constraint $\Cdef{\limpliedby}{v_{\node_\top}}{1}{l}{\bestval-1}$ and the solution improving constraint $\objective \leq \bestval - 1$, followed by a deletion by $\bestval$ results in the desired constraint.
\end{proof}

%%%% -----------------------------------------------------

\paragraph{Deriving the defining constraints}
We will now turn our attention to proving the defining constraints of the variables representing the nodes in the MDD encoding.

\begin{proposition}\label{prop:mainproofmdd}
	Let $\node = \mddOf{k}{l}{u}$ with children $\{\child{\node}{\blocklit_i} \mid i \in \litspartition_{k}\} \cup \{\child{\node}{\lfalse}\}$. 
	
	If the defining constraints for all children are given, then it is possible to derive the defining constraints for $\node$ in the VeriPB proof system. 
\end{proposition}

\begin{proof}	
%	Deriving any of the two defining constraints can be done with a single application of the redundance rule. 
%	The difficult lies in deriving them simultaneously. 
	The first step is to introduce two reification variables $v_\node$ and $v_\node'$ as 
	\begin{align}
		v_\node &\lequiv  \sum_{j \geq k} \sum_{i \in \litspartition_j} \cost_i \blocklit_i \leq l  &&&
		v_\node' &\lequiv \sum_{j \geq k} \sum_{i \in \litspartition_j} \cost_i \blocklit_i \leq u, \label{defvprime} 
	\end{align}
	The next step is showing that these two variables are actually the same; the difficult direction is showing that $v_\node'$ implies $_\node$. 
	We do this by deriving 
	\begin{align}
	&\olnot{\blocklit_m} + \olnot{v_\node'} + v_\node \geq 1 \text{, for every } m \in \litspartition_{k}  \text{, and} \label{caseVarTrue} \\
	&\sum_{i \in \litspartition_{k}} \blocklit_i + \olnot{v_\node'} + v_\node \geq 1 \label{caseVarFalse}.
	\end{align}
	Indeed, \eqref{caseVarTrue} expresses that if $\blocklit_m$ is true, then $v_\node'$ implies $v_\node$; \eqref{caseVarFalse} expresses that if $\blocklit_i$ is false for every $i \in \litspartition_{k}$ , then $v_\node'$ implies $v_\node$. 
	Once these constraints are derived, adding them all up and dividing by $|\litspartition_k|+1$ results in $\olnot{v_\node'} + v_\node \geq 1$, which can be multiplied by $u+1$ and added to $\Creif{\limpliedby}{v_\node'}$ to obtain $v_\node \limpliedby \sum_{i\geq k}^n \cost_i \blocklit_i \leq u$.

	\paragraph{Deriving \eqref{caseVarTrue}} % Claim 1
	We will now proceed with deriving \eqref{caseVarTrue} for variable $\blocklit_m$ with $m \in \litspartition_{k}$. 
	For that reason, consider the node $\node_c = \child{\node}{\blocklit_m} = \mddOf{k+1}{l_c}{u_c}$, for which the defining constraints $\Cdef{\limplies}{v_{\node_c}}{k}{l}{u}$ and $\Cdef{\limpliedby}{v_{\node_c}}{k}{l}{u}$
	are already derived. 
	Assume for now that $\node_c$ is not the true or false leaf; the case for true or false leaf will be handled later.
	The proof is by contradiction, so assume $\neg (\olnot \blocklit_m + \olnot v_\node' + v_\node \geq 1)$ holds, from which we can easily derive the constraints\footnote{The negation of $ \olnot \blocklit_m + \olnot v_\node' + v_\node \geq 1$ is $ \blocklit_m + v_\node' + \olnot v_\node \geq 3$, adding literal axioms to this constraint results in the constraints \eqref{eq:ass:claim1}.}
	\begin{align}
	\blocklit_m \geq 1, &\quad\text{ and }  &v_\node' \geq 1, &\quad\text{ and } &\olnot{v_\node} \geq 1. \label{eq:ass:claim1}
	\end{align} 
	
	On the one hand, by combining $v_\node' \geq 1$  with $\Creif{\limplies}{v_\node'}$,~\footnote{Multiplying  $v_\node' \geq 1$ by $(\sum_{j \geq k} \sum_{i \in \litspartition_j} \cost_i) - l$ added to $\Creif{\limplies}{v_\node'}.$} we can derive
	\begin{align}
	\sum_{j \geq k} \sum_{i \in \litspartition_j} \cost_i \blocklit_i \leq u
	\end{align}
	By adding $\blocklit_m \geq 1$, multiplied by $\cost_m$, we obtain 
	\begin{align}
	\sum_{i \in \litspartition_{k}, i \neq m} \cost_i \blocklit_i + \sum_{j \geq k+1} \sum_{i \in \litspartition_j} \cost_i \blocklit_i \leq u - \cost_m  \label{eq:mdd:casesplitting_var:sufficientfalse}
	\end{align}
	Adding literal axioms ${\blocklit_i} \geq 0$, multiplied by $\cost_i$ for every $i \in \litspartition_{k}$ for which $i \neq m$, we get 
	\begin{align}
	\sum_{j \geq k+1} \sum_{i \in \litspartition_j} \cost_i \blocklit_i \leq u - \cost_m
	\end{align}
	Addition with $\Cdef{\limpliedby}{v_{\node_c}}{k+1}{l_c}{u_c}$ gives
	\begin{align}
	(u_c+1)v_{\node_c} \geq u_c + v_m - u + 1
	\end{align}
	Using the fact that, according to Proposition~\ref{prop:mddcreation}, $u \leq u_c + \cost_m$,  division by a large large enough number (and rounding up) gives 
	\begin{align}
		v_{\node_c} \geq 1 \label{caseVarTrue-vc-true}
	\end{align}

	On the other hand, by combining $\olnot v_\node \geq 1$ with $\Creif{\limpliedby}{v_\node}$~\footnote{Multiplying $\olnot v_\node \geq 1$ by $l+1$ and adding the result with $\Creif{\limpliedby}{v_\node}$.} gives
	\begin{align}
	\sum_{j \geq k} \sum_{i \in \litspartition_j} \cost_i \blocklit_m \geq l + 1 \label{eq:mdd:casesplitting_var:sufficienttrue}
	\end{align}
	By adding $\blocklit_m \geq 1$, multiplied to $\cost_m$ to  $\amo({\node})$ and adding literal axioms $b_i \geq 0$, multiplied by $\cost_i$, for $i \neq j$, we can derive
	\begin{align}
		\blocklit_j \leq 0, \text{ for every } j \neq m \label{eq:otherlitsfalse}
	\end{align}
	By adding these constraints multiplied by $\cost_j$ to \eqref{eq:mdd:casesplitting_var:sufficienttrue} we obtain 
	\begin{align}
	\cost_m \blocklit_m + \sum_{j \geq k+1} \sum_{i \in \litspartition_j} \cost_i \blocklit_m \geq l + 1
	\end{align}
	Now, adding literal axiom $\olnot \blocklit_m \geq 0$ multiplied by $\cost_m$ (and using that $\blocklit_m + \olnot \blocklit_m = 1$) gives
	\begin{align}
	\sum_{j \geq k+1} \sum_{i \in \litspartition_j} \cost_i \blocklit_m \geq l - \cost_m + 1
	\end{align}
	Finally, adding $\Cdef{\limplies}{v_{\node_c}}{k+1}{l_c}{u_c}$ gives
	\begin{align}
		((\sum_{j \geq k+1} \sum_{i \in \litspartition_j} v_i ) - l_c) \olnot v_{\node_c} \geq l  - \cost_m + 1 - l_c
	\end{align}
	Since, according to Proposition~\ref{prop:mddcreation}, $l \geq l_c + \cost_m$, dividing by a large enough number (and rounding up) gives 
	\begin{align}
		\olnot v_{\node_c} \geq 1  \label{caseVarTrue-vc-false}
	\end{align}
	Addition of \eqref{caseVarTrue-vc-true} and \eqref{caseVarTrue-vc-false} results in a contradiction.
	
	Note that in case $\node_c$ is the false leaf, $u - \cost_m < 0$ by design, and therefore \eqref{eq:mdd:casesplitting_var:sufficientfalse} is a contradiction. Similarly, if $\node_c$ is the true leaf, deriving constraint \eqref{eq:mdd:casesplitting_var:sufficienttrue} is sufficient to derive a contradiction. 
	
	\paragraph{Deriving \eqref{caseVarFalse}} % Claim 2
	Let us now derive the constraint \eqref{caseVarFalse} for the outgoing edge encoding the case where no literal in $\{\blocklit_i \mid k \leq i \leq l\}$ is true. In that case, there is a child $\node_\lfalse = \child{\node}{\lfalse} = \mddOf{k+1}{l_\lfalse}{u_\lfalse}$ for which variable $v_{\node_\lfalse}$ is defined by its defining constraints. Note that $\node_\lfalse$ cannot be the true leaf, since then $\node$ is the true leaf as well. We will show how to handle the case where $\node_\lfalse$ is the false-leaf, so for now, assume that $\node_\lfalse$ is neither the true, nor the false leaf. 
	
	In analogy to deriving \eqref{caseVarTrue}, the proof is by contradiction, so assume $\lneg (\sum_{i \in \litspartition_{k}} \blocklit_i + \olnot{v_\node'} + v_\node \geq 1)$. From this assumption, we can derive 
	\begin{align}
		\olnot{\blocklit_j} \geq 1 \quad (\forall j \in \litspartition_{k}),  &\quad\text{ and }  &v_\node' \geq 1, &\quad\text{ and } &\olnot{v_\node} \geq 1. \label{eq:ass:claim2}
	\end{align}
	
	By combining the constraints $\olnot v_\node \geq 1$ and $\Creif{\limpliedby}{v_\node'}$,\footnote{Multiplying $\olnot v_\node \geq 1$ by $u+1$ and adding it to $\Creif{\limpliedby}{v_\node'}$.} we can derive 
	\begin{align}
		\sum_{j \geq k} \sum_{i \in \litspartition_j} \cost_i \blocklit_i \geq l + 1
	\end{align}
	Addition of the constraints $\olnot{\blocklit_i} \geq 1$ (derived in \eqref{eq:ass:claim2}) multiplied by $\cost_i$ for every $i \in \litspartition_{k}$, and using the fact that $\blocklit_i + \olnot{\blocklit}_i = 1$ gives
	\begin{align}
		\sum_{j \geq k+1} \sum_{i \in \litspartition_j} \cost_i \blocklit_i \geq l + 1 \label{eq:mdd:casesplitting_false:sufficientfalse}
	\end{align}
	Now, if $\node_\lfalse$ is the false leaf, we know that $\sum_{j \geq k+1} \sum_{i \in \litspartition_j} \cost_i \blocklit_i \leq l$, which is derived by adding the constraints $\amoUB(\litspartition_j)$ for every $j \geq k+1$. Addition of this constraint with \eqref{eq:mdd:casesplitting_false:sufficientfalse} results in contradiction.
	If however $\node_\lfalse$ is not the false leaf, then we can add \eqref{eq:mdd:casesplitting_false:sufficientfalse} with $\Cdef{\limplies}{v_{\node_\lfalse}}{k+1}{l_\lfalse}{u_\lfalse}$, which results in
	\begin{align}
	((\sum_{j \geq k+1} \sum_{i \in \litspartition_{k}} \cost_i) - l_\lfalse ) \olnot{v_{\node_\lfalse}} \geq l - l_\lfalse + 1
	\end{align}
	From Proposition~\ref{prop:mddcreation} we know that $l \geq l_\lfalse$, so division by a large number gives 
	\begin{align}
	\olnot{v_{\node_\lfalse}} \geq 1 \label{eq:casevarfalse:vcfalse}
	\end{align}

	%TODO: probably put back, since the other direction will derive for true leaf.
	We will now show that we can also derive the constraint $v_{\node_\lfalse} \geq 1$ which can then be added to \eqref{eq:casevarfalse:vcfalse} to derive a contradiction. 
	We do this as follows.
	By combining the constraints $v_\node' \geq 1$ and $\Creif{\limplies}{v_\node'}$,~\footnote{Multiplying $v_\node' \geq 1$ by $(\sum_{j \geq k} \sum_{i \in \litspartition_j} \cost_i) - u$ and adding the result to $\Creif{\limplies}{v_\node'}$} we can derive 
	\begin{align}
	\sum_{j \geq k} \sum_{i \in \litspartition_j} \cost_i \blocklit_i \leq u
	\end{align}
	From addition of literal axioms $\blocklit_i \geq 0$ multiplied by $\cost_i$ for all $i \in \litspartition_{k}$, we can get
	\begin{align}
	\sum_{j \geq k+1} \sum_{i \in \litspartition_j} \cost_i \blocklit_i \leq u
	\end{align}
	Addition with the constraint $\Cdef{\limpliedby}{v_{\node_\lfalse}}{k+1}{l_\lfalse}{u_\lfalse}$, we get
	\begin{align}
		(u_\lfalse+1) v_{\node_\lfalse} \geq u_\lfalse + 1 - u
	\end{align}
	Given that $u_\lfalse \geq u$ (see Proposition~\ref{prop:mddcreation}), division by $u_\lfalse + 1$ results in 
	\begin{align}
	v_{\node_\lfalse} \geq 1 \label{eq:casevarfalse:vctrue}
	\end{align}
	This concludes the proof.
\end{proof}

Finally, it can happen that due to reducedness of MDDs, a node actually represents multiple constraints. 
The following proposition shows that also in that case the required defining constraints can be derived.
\begin{proposition}
	Consider a node $\node = \mddOf{k}{l}{u}$ and let $k'$ and $l'$ be numbers with $k' < k$ and $l' = l + \sum_{j=k}^{k'-1} \max(\{\cost_i \mid i \in \litspartition_j\})$ such that the Boolean functions $\sum_{j \geq k'} \sum_{i \in \litspartition_j} \cost_i \blocklit_i \leq u$ and $\sum_{j \geq k} \sum_{i \in \litspartition_j} \cost_i \blocklit_i \leq u$ are equivalent.
	
	From the defining constraints of $\mddOf{k}{l}{u}$, there is a cutting planes derivation of length XYZ that derives the defining constraints for $\mddOf{k'}{l'}{u}$ for reification variable $v_\node$ from the defining constraints for $\mddOf{k}{l}{u}$.
\end{proposition}

\begin{proof}
	Let us start by deriving $\Cdef{\limplies}{v_\node}{k'}{l'}{u}$. 
	Addition of the constraints $\amoUB(\litspartition_j) = \sum_{i \in \litspartition_j} \cost_i \blocklit_i \leq \max(\{\cost_i \mid i \in \litspartition_j\})$, for every $j$ such that $k' \leq j < k$  gives
	\begin{align}
		\sum_{j = k'}^{k-1} \sum_{i \in \litspartition_j} \cost_i \blocklit_i \leq \sum_{j = k'}^{k-1} \max(\{\cost_i \mid i \in \litspartition_j\})
	\end{align}
	Adding this with $\Cdef{\limplies}{v_\node}{k'}{l'}{u}$ gives
	\begin{align}
		((\sum_{j \geq k} \sum_{i \in \litspartition_j} \cost_i) - l) &v_\node - \sum_{j \geq k'} \sum_{i \in \litspartition_j} \cost_i \blocklit_i \\ 
		& \geq -l - \sum_{j = k'}^{k-1} \max(\{\cost_i \mid i \in \litspartition_j\})
	\end{align}
	By now adding the literal axiom $v_\node \geq 0$ multiplied by $\sum_{j = k'}^{k-1} \max(\{\cost_i \mid i \in \litspartition_j\})$ we obtain in $\Cdef{\limplies}{v_\node}{k}{l}{u}$.
	
	To derive $\Cdef{\limpliedby}{v_\node}{k'}{l'}{u}$, consider constraint $\Cdef{\limpliedby}{v_\node}{k}{l}{u}$, which is
	\begin{align}
		(u+1) v_\node + \sum_{j \geq k} \sum_{i \in \litspartition_j} \cost_i \blocklit_i \geq u+1
	\end{align}
	Adding literal axioms $\blocklit_i \geq 0$, multiplied by $\cost_i$, for every $i \in \litspartition_j$ for every $j$ such that $k' \leq j < k$, we get 
	\begin{align}
	(u+1) v_\node + \sum_{j \geq k'} \sum_{i \in \litspartition_j} \cost_i \blocklit_i \geq u+1
	\end{align}
	which is indeed equal to $\Cdef{\limpliedby}{v_\node}{k'}{l'}{u}$.
\end{proof}

\end{document}